\documentclass[10pt,a4paper]{article}
\usepackage{commath}
\usepackage{graphicx}
\usepackage{enumitem}
\usepackage{float}
\usepackage{amsthm,amssymb,mathrsfs,amsmath}

\usepackage[a4paper,bindingoffset=0.2in,left=0.5in,right=0.5in,top=1in,bottom=1in,footskip=.25in]{geometry}
\usepackage{cite}
\usepackage[colorlinks,citecolor=blue,urlcolor=blue,bookmarks=false,hypertexnames=true]{hyperref}
\hypersetup{citecolor=blue}
\newtheorem{definition}{Definition}[section]
\newtheorem{theorem}{Theorem}[section]
\newtheorem{corollary}{Corollary}[theorem]
\newtheorem{lemma}[theorem]{Lemma}

\newtheorem{remark}{Remark}

\title{Short time quaternion quadratic phase Fourier transform and its uncertainty principles}
\author{Bivek Gupta$^a$\thanks{$^a$bivekgupta040792@gmail.com}, Amit K. Verma$^b$\thanks{$^b$akverma@iitp.ac.in}, \\\small{\it{$^{a,b}$ Department of Mathematics,}} \\\small{\it{Indian Institute of Technology Patna,}}\\\small{\it{ Bihta, Patna 801103, (BR) India.}}}
\date{\today}

\begin{document}
\maketitle
\begin{abstract}
In this paper, we extend the quadratic phase Fourier transform of a complex valued functions to that of the quaternion valued functions of two variables. We call it the quaternion quadratic phase Fourier transform (QQPFT). Based on the relation between the QQPFT and the quaternion Fourier transform (QFT) we obtain the sharp Hausdorff-Young inequality for QQPFT. We define the short time quaternion quadratic phase Fourier transform (STQQPFT) and explore some of its properties including inner product relation and inversion formula. We find its relation with that of the 2D quaternion ambiguity function and the quaternion Wigner-Ville distribution associated with QQPFT and obtain the Lieb's uncertainty and entropy uncertainty principles for these three transforms.
\end{abstract}

\section{Introduction}
In \cite{castro2014quadratic},\cite{castro2018new} authors have studied the quadratic phase Fourier transform (QPFT) defined as 
\begin{equation}\label{P5eqn1}
(\mathcal{Q}^{\wedge}f)(\xi)=\int_{\mathbb{R}}\frac{1}{\sqrt{2\pi}}e^{i\left(At^2+Bt\xi+C\xi^2+Dt+E\xi\right)}f(t)dt,~\xi\in\mathbb{R},
\end{equation}
where $f\in L^2(\mathbb{R},\mathbb{C}),$ $\wedge=(A,B,C,D,E),~B\neq 0$
which generalizes the classical Fourier transform (FT). Several other important integral transforms like fractional Fourier transform (FrFT) \cite{almeida1994fractional},\cite{namias1980fractional}, linear canonical transform (LCT), Fresnel transform and Lorentz transform  can be obtained by choosing $\wedge$ appropriately and amplifying \eqref{P5eqn1} with suitable constants. Along with several important properties like Riemann-Lebesgue lemma, Plancherel theorem, authors in \cite{castro2018new} have given several convolution and obtained the convolution theorem associated with the QPFT. Recently, Shah et al. \cite{shah2021uncertainty} generalized several uncertainty principles for the FT, FrFT (\cite{verma2021note}) and LCT for the QPFT defined in \eqref{P5eqn1}.  Even though QPFT generalizes several integral transforms as mentioned above, but due to the presence of global kernel it fails in giving the local quadratic phase spectrum content of non-transient signals. To overcome this, Shah et al. (\cite{shah2021short}) formulated a short time quadratic phase Fourier transform (STQPFT) and studied its important properties. They have generalized the  Heisenberg's, logarithmic and local uncertainty principles (UPs)  for FT and fractional FT (\cite{verma2021note},\cite{wilczok2000new}) and Lieb's UP for short time FT (\cite{grochenig2001foundations}) in the context of STQPFT. Apart from STQPFT, wavelet transform and Wigner-Ville distribution associated with the QPFT has also been studied. Shah et al.\cite{shah2022quadratic} proposed a novel quadratic phase Wigner distribution by combining the advantages of Wigner distribution and the QPFT. They obtained several fundamental properties including Moyel's formula and inversion formula. Prasad et al. \cite{prasad2020quadratic} defined the wavelet transform associated with the QPFT, and studied its properties like inversion formula, Parseval's formula and also its continuity on some function spaces.

In 1843 W.R. Hamilton first introduced the quaternion algebra. It is denoted by $\mathbb{H}$ in his honor. In Harmonic analysis and applied mathematics, the FT is an essential tool so its extension to the quaternion valued functions has become an interesting problem. The quaternion Fourier transform (QFT) was introduced by Ell \cite{ell1993quaternion}  for the analysis of $2D$ linear time-invariant partial differential system and later applied it in color image processing\cite{ell2006hypercomplex}. In the analysis of quaternion valued functions quaternion Fourier transform plays a significant role. Because of the non-commutativity of the quaternion multiplication, the Fourier transform of the quaternion valued function on $\mathbb{R}^2$ can be classified into various types, viz., right-sided, left-sided and two-sided Fourier transform \cite{bahri2014continuous},\cite{bahri2008uncertainty},\cite{ell1993quaternion}. Cheng et. al \cite{cheng2019plancherel} gave the inversion theorem and the Plancherel theorem for the right sided QFT, and also obtained its relation  with the left sided and the two sided QFT for the quaternion valued square integrable functions. It transforms a quaternion valued $2D$ signal into a quaternion valued frequency domain signal. 

Lian \cite{lian2018uncertainty}, proved various inequalities like Pitt's inequality, logarithmic UP using the method adopted by Beckner \cite{beckner1995pitt} in the case of  complex variables, entropy UP without using the sharp Hausdorff-Young inequality, for the two-sided QFT with optimal constants, which are same to those obtained in the complex case. The logarithmic UP obtained in \cite{lian2018uncertainty} is different from that given in \cite{chen2015pitt}. In \cite{lian2020sharp}, author obtained the sharp Hausdorff-Young inequality, using the orthogonal plan split of the quaternion \cite{hitzer2013orthogonal}, for the two sided QFT followed by the Hirschman's entropy UP using the standard differential approach. In \cite{lian2020sharpAF}, author has extended the QFT to the Clifford valued function defined on $\mathbb{R}^n,$ namely geometric FT, and derived several sharp inequalities including sharp Hausdorff-Young inequality and sharp Pitt's inequality, followed by the sharp entropy inequality for the Clifford ambiguity functions. Recently, QFT has been extended to the quaternion fractional Fourier transform (QFrFT) and quaternion linear canonical transform (QLCT).

Replacing the kernels $\mathcal{K}^i(t_1,\xi_1)=\frac{1}{\sqrt{2\pi}} e^{-it_1,\xi_1}$ and $\mathcal{K}^j(t_2,\xi_2)=\frac{1}{\sqrt{2\pi}}e^{-jt_2,\xi_2},$ in the definition 
\begin{align}\label{P5eqn2}
(\mathcal{F}_{\mathbb{H}}f)(\boldsymbol \xi)=\int_{\mathbb{R}^2}\mathcal{K}^i(t_1,\xi_1)f(\boldsymbol{t})\mathcal{K}^j(t_2,\xi_2)d\boldsymbol t,~\boldsymbol\xi=(\xi_1,\xi_2)\in\mathbb{R}^2,
\end{align}
of the two sided QFT (\cite{lian2021quaternion}), with that of the kernels  of the FrFT (\cite{namias1980fractional},\cite{almeida1994fractional},\cite{verma2021note}) and LCT, respectively, results in the two-sided quaternion fractional Fourier transform (QFrFT) and the two sided quaternion linear canonical transform (QLCT)\cite{kou2016uncertainty}. Analogously, the right side and the left sides QFrFT and QLCT have been defined in the literature (see \cite{wei2013different}, \cite{kou2016uncertainty} ).
Kou et al. \cite{kou2016uncertainty} adopted the approach by Chen et al. \cite{chen2015pitt} to obtain the energy theorem and proved the Heisenberg's UP for the QLCT. Using the orthogonal plan split method, authors in \cite{kundu2022uncertainty} have obtained the relation of the two-sided QLCT with that of the LCT and obtained some important inequalities and uncertainty principles of two-sided QLCT.

Bahri et al.\cite{bahri2010windowed} generalized the classical windowed Fourier transform to quaternion valued functions of two variables. Using the machinery of the right sided QFT \cite{bahri2008uncertainty}, authors proved several important properties including reconstruction formula, reproducing kernel and orthogonality relation. Following the methods adopted by Wilczok \cite{wilczok2000new}, they also obtained the Heisenberg UP for the QWFT. In \cite{bahri2020uncertainty}, authors gave the alternate proofs of the properties studied in \cite{bahri2010windowed}. They also studied the Pitt's inequality, Lieb's inequality and the logarithmic UP for the two sided QWFT studied in \cite{bahri2010windowed}. Including, the orthogonality property, authors in \cite{kamel2019uncertainty},\cite{brahim2020uncertainty} studied the local UP, logarithmic UP, Beckner's UP in terms of entropy, Lieb's UP, Amrein-Berthier UP for the two sided QWFT. Replacing the Fourier kernel in the left sided, right sided or two sided QWFT by the kernels of the FrFT (or LCT), results in the left sided, right sided and two sided QWFrFT (or QWLCT) respectively. In \cite{fu2012balian} authors have studied the two sided QWFT with the real valued window function and studied its important properties and the associated Balian-Low theorem. In \cite{gao2020quaternion}, authors studied the orthogonality relation along with the Heisenberg's UP for the two sided QWLCT, with quaternion valued window function. Bahri, in \cite{bahri2014two}, has extended the classical ambiguity function (AF) and the Wigner-Ville distribution (WVD) to the quaternion algebra setting , namely, quaternion ambiguity function (QAF) and quaternion Wigner-Ville distribution (QWVD). They studied several important properties including Moyel's principle and reconstruction formula for these two sided QAF and QWVD. Authors in \cite{fan2017quaternion} have extended these two sided QAF and QWVD in the linear canonical domain and obtained the relation among them. They have also studied their important properties like shifting, dilation, reconstruction formula, Moyal's theorem, etc.

Several important properties along with the UPs of the QPFT along with the STQPFT have been studied for the function of complex variables as mentioned above. The QPFT has more degree of freedom and is more flexible with the parameters involved than the FT, FrFT and the LCT, with the same computational cost as the FT, it is natural to extend QPFT to quaternion setting. To the best of our knowledge none of the QPFT and the STQPFT have been explored for the quaternion valued functions. Due to non-commutativity of the quaternion multiplication we can define at least three different types of quaternion quadratic phase Fourier transform (QQPFT), viz., right-sided, left-sided and two-sided. In this article, we concentrate on the two-sided QQPFT and based on its relation with the quaternion Fourier transform (QFT) we obtain the sharp Hausdorff-Young inequality using which we give the R\`enyi and Shannon entropy UP for QQPFT.  We also define the STQQPFT and explore its important properties like, boundedness, linearity, translation, scaling, inner product relation and inversion formula. Based on the sharp Hausdorff-Young inequality we obtain the Lieb's uncertainty and entropy uncertianty principles of the STQQPFT followed by the same for the newly defined $2D$ quaternion quadratic phase ambiguity function (QQPAF) and $2D$ quaternion quadratic phase Wigner-Ville distribution (QQPWVD), using the relation of the later transforms with that of the STQQPFT.

The organization of the paper is as follows:
In section 2, we recall some basic definitions and properties of quaternion algebra. In section 3, we give the definition of two sided QQPFT and study its important properties, like Parseval's identity, sharp Hausdorff-Young inequality, R\`enyi and Shannon entropy UPs. In section 4, we have defined the two sided STQQPFT and studied its properties and its relations with that of the proposed two sided QQPAF and the QQPWVD, based on which we obtain the Lieb's and entropy UPs for these three transforms. Finally, in section 5, we conclude our paper.

\section{Preliminaries}
The field of real and complex numbers are respectively denoted by $\mathbb{R}$ and $\mathbb{C}.$ Let 
$$\mathbb{H}=\{r=r_0+ir_1+jr_2+kr_3 :r_0,r_1,r_2,r_3\in\mathbb{R}\},$$ where $i,j$ and $k$ are the imaginary units such that they satisfy the following Hamilton's multiplication rule
$$ij=k=-ji,~jk=i=-kj,~ki=j=-ik,~i^2=j^2=k^2=1.$$

For a quaternion $r=r_0+ir_1+jr_2+kr_3,$ we call $r_0$ the real scalar part of $r,$ and denote it by $Sc(r).$ The scalar part satisfies the following cyclic multiplication symmetry (\cite{hitzer2007quaternion})
\begin{align}\label{P5eqn3}
Sc(pqr)=S(qrp)=Sc(rpq),~\forall~p,q,r\in\mathbb{H}.
\end{align}

We denote the quaternion conjugate of $r$ as $\bar{r}$ and is defined as
$$\bar{r}=r_0-ir_1-jr_2-kr_3.$$
The quaternion conjugate satisfy the following
\begin{align}\label{P5eqn4}
\overline{qr}=\bar{r}\bar{q},~\overline{q+r}=\bar{q}+\bar{r},~\bar{\bar{q}}=q,~\forall~q,r\in\mathbb{H}.
\end{align}
The modulus of $r\in\mathbb{H}$ is defined as 
\begin{align}\label{P5eqn5}
|r|=\sqrt{r\bar{r}}=\left(\sum_{l=0}^3r_l^2\right)^{\frac{1}{2}},
\end{align}
and it satisfies $|qr|=|q||r|,~\forall~q,r\in\mathbb{H}.$

A quaternion valued function $h$ defined on $\mathbb{R}^n$  can be written as
\begin{align*}
h(\boldsymbol x)=h_0(\boldsymbol x)+ih_1(\boldsymbol x)+jh_2(\boldsymbol x)+kh_3(\boldsymbol x),~\boldsymbol x\in\mathbb{R}^n,
\end{align*}
where $h_0,h_1,h_2$ and $h_3$ are real valued function on $\mathbb{R}^n.$

If $1\leq q<\infty,$ then the $L^q-$norm of $h$ is defined by
\begin{align}\label{P5eqn6}
\|h\|_{L^q_\mathbb{H}(\mathbb{R}^n)}
&=\left(\int_{\mathbb{R}^n}|h(\boldsymbol x)|^qd\boldsymbol x\right)^\frac{1}{q}\notag\\
%&=\left\{\int_{\mathbb{R}^n}\left(|h_0(\boldsymbol x)|^2+|%h_1(\boldsymbol x)|^2+|h_2(\boldsymbol x)|^2+|h_3(\boldsymbol x)|%^2\right)^\frac{q}{2}d\boldsymbol x\right\}^\frac{1}{q}\notag\\
&=\left\{\int_{\mathbb{R}^n}\left(\sum_{l=0}^3|h_l(\boldsymbol x)|^2\right)^\frac{q}{2}d\boldsymbol x\right\}^\frac{1}{q}
\end{align}
and $L^q_\mathbb{H}(\mathbb{R}^n)$ is a Banach space of all measurable quaternion valued functions $f$ having finite $L^q-$norm.
$L^\infty_\mathbb{H}(\mathbb{R}^n)$ is the set of all essentially bounded quaternion valued measurable functions with norm 
\begin{align}\label{P5eqn7}
\|f\|_{L^\infty_\mathbb{H}(\mathbb{R}^n)}=\mbox{ess~sup}_{\boldsymbol x\in\mathbb{R}^n}|f(\boldsymbol x)|.
\end{align}
Moreover, the quaternion valued inner product 
\begin{align}\label{P5eqn8}
(f,g)=\int_{\mathbb{R}^n}f(\boldsymbol x)\overline{g(\boldsymbol x)}d\boldsymbol x,
\end{align}
with symmetric real scalar part
\begin{align}\label{P5eqn9}
\langle f,g\rangle
&=\frac{1}{2}[(f,g)+(g,f)]\notag\\
&=\int_{\mathbb{R}^n}Sc\left[f(\boldsymbol x)\overline{g(\boldsymbol x)}\right]d\boldsymbol x\notag\\
&=Sc\left(\int_{\mathbb{R}^n}f(\boldsymbol x)\overline{g(\boldsymbol x)}d\boldsymbol x\right)
\end{align}
turns $L^2_\mathbb{H}(\mathbb{R}^n)$ to a Hilbert space, where the norm in equation \eqref{P5eqn6} can be expressed as 
\begin{align}
\|f\|_{L^2_\mathbb{H}(\mathbb{R}^n)}=\sqrt{\langle f,f \rangle}=\sqrt{(f,f)}=\left(\int_{\mathbb{R}^n}|f(\boldsymbol x)|^2d\boldsymbol x\right)^\frac{1}{2}.
\end{align}

\section{Quaternion quadratic phase Fourier transform (QQPFT)}
In this section we give a definition of quaternion quadratic phase Fourier transform (QQPFT) and study its important properties.
\begin{definition}\label{P5Defn3.1}
Let $\wedge_l=(A_l,B_l,C_l,D_l,E_l),A_l,B_l,C_l,D_l,E_l\in\mathbb{R} ~\mbox{and}~B_l\neq 0~\mbox{for}~l=1,2$. The quaternion quadratic phase Fourier transform (QQPFT) of $f(\boldsymbol t)\in L^2_{\mathbb{H}}(\mathbb{R}^2),~\boldsymbol t=(t_1,t_2),$ is defined by 
\begin{align}\label{P5eqn11}
(\mathcal{Q}^{\wedge_1,\wedge_2}_{\mathbb{H}} f)(\boldsymbol\xi)=\int_{\mathbb{R}^2}\mathcal{K}^i_{\wedge_1}(t_1,\xi_1)f(\boldsymbol t)\mathcal{K}^j_{\wedge_2}(t_2,\xi_2)d\boldsymbol t,~\boldsymbol{\xi}=(\xi_1,\xi_2)\in\mathbb{R}^2
\end{align}
where 
\begin{align}\label{P5eqn12}
\mathcal{K}^i_{\wedge_1}(t_1,\xi_1)=\frac{1}{\sqrt{2\pi}}e^{-i\left(A_1t_1^2+B_1t_1\xi_1+C_1\xi_1^2+D_1t_1+E_1\xi_1\right)}
\end{align}
and 
\begin{align}\label{P5eqn13}
\mathcal{K}^i_{\wedge_2}(t_2,\xi_2)=\frac{1}{\sqrt{2\pi}}e^{-j\left(A_2t_2^2+B_2t_2\xi_2+C_2\xi_2^2+D_2t_2+E_2\xi_2\right)}.
\end{align}
\end{definition}
The corresponding inversion formula is given by
\begin{align}\label{P5eqn14}
f(\boldsymbol t)=|B_1B_2|\int_{\mathbb{R}^2}\overline{\mathcal{K}^i_{\wedge_1}(t_1,\xi_1)}(\mathcal{Q}^{\wedge_1,\wedge_2}_{\mathbb{H}} f)(\boldsymbol\xi)\overline{\mathcal{K}^i_{\wedge_2}(t_2,\xi_2)}d\boldsymbol\xi
\end{align}

\subsection{Relation between QQPFT and QFT}
We now see an important relation between the QQPFT and the QFT, which plays a vital role in obtaining the sharp Hausdorff-Young inequality for the QQPFT.
\begin{align*}
(\mathcal{Q}^{\wedge_1,\wedge_2}_{\mathbb{H}} f)(\boldsymbol\xi)&=\frac{1}{2\pi}\int_{\mathbb{R}^2}e^{-i\left(A_1t_1^2+B_1t_1\xi_1+C_1\xi_1^2+D_1t_1+E_1\xi_1\right)}f(\boldsymbol t)e^{-j\left(A_2t_2^2+B_2t_2\xi_2+C_2\xi_2^2+D_2t_2+E_2\xi_2\right)}d\boldsymbol t\\
&=e^{-i\left(C_1\xi^2+E_1\xi_1\right)}\left\{\frac{1}{2\pi}\int_{\mathbb{R}^2} e^{-iB_1t_1\xi_1}\tilde{f}(\boldsymbol t)e^{-jB_2t_2\xi_2}d\boldsymbol t\right\}e^{-j\left(C_2\xi^2+E_2\xi_2\right)},
\end{align*}
where
\begin{align}\label{P5eqn15}
\tilde{f}(\boldsymbol t)=e^{-i\left(A_1t_i^2+D_1t_1\right)}f(\boldsymbol t)e^{-j\left(A_2t_i^2+D_2t_2\right)}.
\end{align}
Thus,
\begin{align}\label{P5eqn16}
(\mathcal{Q}^{\wedge_1,\wedge_2}_{\mathbb{H}} f)(\boldsymbol\xi)&=e^{-i\left(C_1\xi^2+E_1\xi_1\right)}\left(\mathcal{F}_{\mathbb{H}}\tilde{f}\right)(B_1\xi_1,B_2\xi_2)e^{-j\left(C_2\xi^2+E_2\xi_2\right)}
\end{align}
where
\begin{align}\label{P5eqn17}
\left(\mathcal{F}_{\mathbb{H}}\tilde{f}\right)(\boldsymbol \xi)=\int_{\mathbb{R}^2}\frac{1}{\sqrt{2\pi}} e^{-it_1\xi_1}\tilde{f}(\boldsymbol t)\frac{1}{\sqrt{2\pi}}e^{-jt_2\xi_2}d\boldsymbol t.
\end{align}
Based on this relation between QQPFT and the QFT, we obtain the following important inequality.
\begin{theorem}\label{P5Theo3.1}
(Sharp Hausdorff-Young Inequality): Let $1\leq p\leq 2,$ $\frac{1}{p}+\frac{1}{q}=1$ and $f\in L^2_\mathbb{H}(\mathbb{R}^2),$ then
\begin{align}\label{P5eqn18}
\|\mathcal{Q}^{\wedge_1,\wedge_2}_{\mathbb{H}} f\|_{L^q_\mathbb{H}(\mathbb{R}^2)}\leq \frac{(2\pi)^{\frac{1}{q}-\frac{1}{p}}A_p^2}{|B_1B_2|^\frac{1}{q}}\|f\|_{L^p_\mathbb{H}(\mathbb{R}^2)},
\end{align}
\end{theorem}
where $A_p=\left(\frac{p^{\frac{1}{p}}}{q^{\frac{1}{q}}}\right)^\frac{1}{2}.$
\begin{proof}
Using the relation between the QQPFT and the QFT, we get
\begin{align*}
\|\mathcal{Q}^{\wedge_1,\wedge_2}_{\mathbb{H}} f\|_{L^q_\mathbb{H}(\mathbb{R}^2)}
&=\left(\int_{\mathbb{R}^2}\left|\left(\mathcal{F}_\mathbb{H}\tilde{f}\right)(B_1\xi_1,B_2\xi_2)\right|^qd\boldsymbol\xi\right)\frac{1}{q}\\
&=\frac{1}{|B_1B_2|^{\frac{1}{q}}}\|\mathcal{F}_\mathbb{H}\tilde{f}\|_{L^q_\mathbb{H}(\mathbb{R}^2)}.
\end{align*}
Using the sharp Hausdorff-Young inequality (\cite{lian2020sharp}) for the QFT, we get
\begin{align*}
\|\mathcal{Q}^{\wedge_1,\wedge_2}_{\mathbb{H}} f\|_{L^q_\mathbb{H}(\mathbb{R}^2)}\leq \frac{(2\pi)^{\frac{1}{q}-\frac{1}{p}}A_p^2}{|B_1B_2|^\frac{1}{q}}\|\tilde{f}\|_{L^p_\mathbb{H}(\mathbb{R}^2)}.
\end{align*}
Substituting $\tilde{f},$ from \eqref{P5eqn15}, we get \eqref{P5eqn18}. This completes the proof.
\end{proof}
\begin{theorem}\label{P5Theo3.2}
(Parseval's formula):
Let $f,g\in L^2_\mathbb{H}(\mathbb{R}^2),$ then 
\begin{align}\label{P5eqn19}
\langle f,g\rangle=|B_1B_2|\langle \mathcal{Q}^{\wedge_1,\wedge_2}_{\mathbb{H}} f, \mathcal{Q}^{\wedge_1,\wedge_2}_{\mathbb{H}} g \rangle.
\end{align}
In particular,
\begin{align}\label{P5eqn20}
\|f\|^2_{L^2_\mathbb{H}(\mathbb{R}^2)}=|B_1B_2|\|\mathcal{Q}^{\wedge_1,\wedge_2}_{\mathbb{H}} f\|^2_{L^2_\mathbb{H}(\mathbb{R}^2\times\mathbb{R}^2)}.
\end{align}
\end{theorem}
\begin{proof}
By the Parseval's formula for the QFT of the function $\tilde{f}$ and $\tilde{g},$ we have
\begin{align*}
\langle \tilde{f},\tilde{g}\rangle&=\langle \mathcal{F}_{\mathbb{H}}\tilde{f},\mathcal{F}_{\mathbb{H}}\tilde{g}\rangle\\
&=Sc \int_{\mathbb{R}^2}|B_1B_2|\left(\mathcal{F}_{\mathbb{H}}\tilde{f}\right)(B_1\xi_1,B_2\xi_2)\overline{\left(\mathcal{F}_{\mathbb{H}}\tilde{g}\right)(B_1\xi_1,B_2\xi_2)}d\boldsymbol \xi.
\end{align*}
Using the relation between the QQPFT and the QFT, we get
\begin{align*}
\langle \tilde{f},\tilde{g}\rangle&=|B_1B_2|\int_{\mathbb{R}^2}Sc\left[e^{i\left(C_1\xi_1^2+E_1\xi_2\right)}\left(\mathcal{Q}^{\wedge_1,\wedge_2}_{\mathbb{H}} g \right)(\boldsymbol\xi)\overline{\left(\mathcal{Q}^{\wedge_1,\wedge_2}_{\mathbb{H}} g\right)(\boldsymbol \xi)}e^{-i\left(C_1\xi_1^2+E_1\xi_2\right)}\right]d\boldsymbol\xi\\
&=|B_1B_2|\int_{\mathbb{R}^2}Sc\left[\left(\mathcal{Q}^{\wedge_1,\wedge_2}_{\mathbb{H}} g \right)(\boldsymbol\xi)\overline{\left(\mathcal{Q}^{\wedge_1,\wedge_2}_{\mathbb{H}} g\right)(\boldsymbol \xi)}\right]d\boldsymbol\xi\\
&=|B_1B_2|\langle \mathcal{Q}^{\wedge_1,\wedge_2}_{\mathbb{H}} f, \mathcal{Q}^{\wedge_1,\wedge_2}_{\mathbb{H}} g \rangle.
\end{align*}
This proof equation \eqref{P5eqn19}. In particular, if we take $f=g,$ in equation \eqref{P5eqn19}, we get equation \eqref{P5eqn20}.
 
This completes the proof.
\end{proof}

\subsection{R\`enyi and Shannon entropy uncertainty principle}
In this subsection we obtain the R\`enyi and Shannon entropy UPs for the proposed QQPFT. Analogous results for the FrFT of complex valued function can be found in \cite{guanlei2009generalized}. Recently, Shannon entropy UP for the QPFT and the two sided QLCT are studied in \cite{shah2021uncertainty} and \cite{kundu2022uncertainty} respectively. Below we prove, R\`enyi UP for the QQPFT and obtain the Shannon UP in limiting case. We start with the following definition.
\begin{definition}\label{P5Defn3.2}
\cite{dembo1991information, guanlei2009generalized}
The R\`enyi entropy of a probability density function $P$ on $\mathbb{R}^n$ is defined by 
\begin{align}\label{P5eqn21}
H_\alpha(P)=\frac{1}{1-\alpha}\log \left(\int_{\mathbb{R}^n}[P(\boldsymbol t)]^\alpha d\boldsymbol t\right),~\alpha>0, \alpha\neq 1.
\end{align}
If $\alpha\rightarrow 1,$ then \eqref{P5eqn21} leads to the following Shannon entropy
\begin{align}\label{P5eqn22}
E(P)=-\int_{\mathbb{R}^n}P(\boldsymbol t)\log [P(\boldsymbol t)]d\boldsymbol t
\end{align}
\end{definition}
\begin{theorem}\label{P5Theo3.3}
If $f\in L^2_\mathbb{H}(\mathbb{R}^2),$ $\frac{1}{2}<\alpha<1$ and $\frac{1}{\alpha}+\frac{1}{\beta}=2,$ then 
\begin{align*}
H_\alpha(|f|^2)+H_\beta\left(\left|\sqrt{|B_1B_2|}\left(\mathcal{Q}^{\wedge_1,\wedge_2}_\mathbb{H}f\right)(\boldsymbol \xi)\right|^{2}\right)\geq-\log(|B_1B_2|)-2\log(2\pi)-\left(\frac{1}{1-\alpha}\log(2\alpha)+\frac{1}{1-\beta}\log(2\beta)\right).
\end{align*}
\end{theorem}
\begin{proof}
By Hausdorff-Young inequality \eqref{P5eqn18}, we have 
\begin{align}\label{P5eqn23}
\left(\int_{\mathbb{R}^2}\left|\left(\mathcal{Q}^{\wedge_1,\wedge_2}_{\mathbb{H}}f\right)(\boldsymbol\xi)\right|^qd\boldsymbol\xi\right)^{\frac{1}{q}}\leq \frac{(2\pi)^{\frac{1}{q}-\frac{1}{p}}A_p^2}{|B_1B_2|^{\frac{1}{q}}}\left(\int_{\mathbb{R}^2}|f(\boldsymbol t)|^pd\boldsymbol t\right)^{\frac{1}{p}}.
\end{align} 
Putting $p=2\alpha$ and $q=2\beta,$ in equation \eqref{P5eqn23}, we have
\begin{align*}
\frac{1}{\sqrt{|B_1B_2|}}\left(\int_{\mathbb{R}^2}\left|\sqrt{|B_1B_2|}\left(\mathcal{Q}^{\wedge_1,\wedge_2}_{\mathbb{H}}f\right)(\boldsymbol\xi)\right|^{2\beta}d\boldsymbol\xi\right)^{\frac{1}{2\beta}}\leq \frac{(2\pi)^{\frac{1}{2\beta}-\frac{1}{2\alpha}}A_{2\alpha}^2}{|B_1B_2|^{\frac{1}{2\beta}}}\left(\int_{\mathbb{R}^2}|f(\boldsymbol t)|^{2\alpha}d\boldsymbol t\right)^{\frac{1}{2\alpha}}.
\end{align*}
This implies
\begin{align}\label{P5eqn24}
\frac{|B_1B_2|^{\frac{1}{\beta}-1}}{(2\pi)^{\frac{1}{\alpha}-\frac{1}{\beta}}A^4_{2\alpha}}\leq \left(\int_{\mathbb{R}^2}|f(\boldsymbol t)|^{2\alpha}d\boldsymbol t\right)^{\frac{1}{\alpha}}\left(\int_{\mathbb{R}^2}\left|\sqrt{|B_1B_2|}\left(\mathcal{Q}^{\wedge_1,\wedge_2}_{\mathbb{H}}f\right)(\boldsymbol\xi)\right|^{2\beta}d\boldsymbol\xi\right)^{-\frac{1}{\beta}}.
\end{align}
Since $\frac{1}{\alpha}+\frac{1}{\beta}=2,$ we have
\begin{align}\label{P5eqn25}
\frac{\alpha}{1-\alpha}=\frac{\beta}{1-\beta}.
\end{align}
Raising to the power $\frac{\alpha}{1-\alpha}$ in \eqref{P5eqn24} and using \eqref{P5eqn25}, we get
\begin{align*}
\frac{|B_1B_2|^{-1}}{(2\pi)^{(\frac{1}{\alpha}-\frac{1}{\beta})(\frac{\alpha}{1-\alpha})}A^{\frac{4\alpha}{1-\alpha}}_{2\alpha}}\leq \left(\int_{\mathbb{R}^2}|f(\boldsymbol t)|^{2\alpha}d\boldsymbol t\right)^{\frac{1}{1-\alpha}}\left(\int_{\mathbb{R}^2}\left|\sqrt{|B_1B_2|}\left(\mathcal{Q}^{\wedge_1,\wedge_2}_{\mathbb{H}}f\right)(\boldsymbol\xi)\right|^{2\beta}d\boldsymbol\xi\right)^{\frac{1}{1-\beta}}.
\end{align*}
Taking $\log$ on both sides, we get
\begin{align}\label{P5eqn26}
-\log(|B_1B_2|)-\log&\left((2\pi)^{(\frac{1}{\alpha}-\frac{1}{\beta})(\frac{\alpha}{1-\alpha})}A^{\frac{4\alpha}{1-\alpha}}_{2\alpha}\right)\notag\\
&\leq \frac{1}{1-\alpha}\log \left(\int_{\mathbb{R}^2}|f(\boldsymbol t)|^{2\alpha}d\boldsymbol t\right)+\frac{1}{1-\beta}\log \left(\int_{\mathbb{R}^2}\left|\sqrt{|B_1B_2|}\left(\mathcal{Q}^{\wedge_1,\wedge_2}_\mathbb{H}\right)(\boldsymbol \xi)\right|^{2\alpha}d\boldsymbol\xi\right).
\end{align}
Thus, it follows that
\begin{align}\label{P5eqn27}
H_\alpha(|f|^2)+H_\beta\left(\left|\sqrt{|B_1B_2|}\left(\mathcal{Q}^{\wedge_1,\wedge_2}_\mathbb{H}f\right)(\boldsymbol \xi)\right|^{2}\right)\geq-\log(|B_1B_2|)-2\log(2\pi)-\left(\frac{1}{1-\alpha}\log(2\alpha)+\frac{1}{1-\beta}\log(2\beta)\right).
\end{align}
This is the R\`enyi entropy UP for QQPFT. 
\end{proof}
\begin{remark}
If $\alpha\rightarrow 1,$ then $\beta\rightarrow 1$ and in this case equation \eqref{P5eqn27} can be written as
\begin{align}\label{P5eqn28}
&E(|f|^2)+E\left(\left|\sqrt{|B_1B_2|}\left(\mathcal{Q}^{\wedge_1,\wedge_2}_\mathbb{H}f\right)(\boldsymbol \xi)\right|^{2}\right)\geq-\log(|B_1B_2|)-2\log(2\pi)+2-\log4\notag,\\
&\hspace{-0.6cm}\mbox{i.e.,}~ E(|f|^2)+E\left(\left|\sqrt{|B_1B_2|}\left(\mathcal{Q}^{\wedge_1,\wedge_2}_\mathbb{H}f\right)(\boldsymbol \xi)\right|^{2}\right)\geq\log \left(\frac{e^2}{16\pi^2|B_1B_2|}\right).
\end{align}
This is the Shannon entropy UP for QQPFT.
\end{remark}

\section{Short time quaternion quadratic phase Fourier transform}
In this section we give the definition of the STQQPFT and study its properties. We obtain its relation with that of the quaternion AF and the quaternion WVD associated with the QQPFT. 
\begin{definition}\label{P5Defn4.1}
Let $\wedge_l=(A_l,B_l,C_l,D_l,E_l),A_l,B_l,C_l,D_l,E_l\in\mathbb{R} ~\mbox{and}~B_l\neq 0~\mbox{for}~l=1,2$. The short time quaternion quadratic phase Fourier transform (STQQPFT) of a function $f\in L^2_\mathbb{H}(\mathbb{R}^2)$ with respect to a quaternion window function (QWF) $g\in L^2_\mathbb{H}(\mathbb{R}^2)\cap L^\infty_\mathbb{H}(\mathbb{R}^2)$ is defined by 
\begin{align*}
\left(\mathcal{S}^{\wedge_1,\wedge_2}_{\mathbb{H},g}f\right)(\boldsymbol x,\boldsymbol\xi)=\int_{\mathbb{R}^2}\mathcal{K}^i_{\wedge_1}(t_1,\xi_1)f(\boldsymbol t)\overline{g(\boldsymbol t-\boldsymbol x)}\mathcal{K}^j_{\wedge_2}(t_2,\xi_2)d\boldsymbol t,~(\boldsymbol x,\boldsymbol\xi)\in\mathbb{R}^2\times\mathbb{R}^2,
\end{align*}
where $\mathcal{K}^i_{\wedge_1}(t_1,\xi_1)$ and $\mathcal{K}^j_{\wedge_2}(t_2,\xi_2)$ are given by equations \eqref{P5eqn12} and \eqref{P5eqn13}, respectively.
\end{definition}
We now derive some of the basic properties of the STQQPFT. But before that we state the following lemma:
\begin{lemma}\label{P5Lemma4.1}
Let $\boldsymbol t=(t_1,t_2),\boldsymbol\xi=(\xi_1,\xi_2),\boldsymbol k=(k_1,k_2)\in\mathbb{R}^2,r\in\mathbb{R}.$ Then the kernel $\mathcal{K}^1_{\wedge_1}(t_1,\xi_1)$ and $\mathcal{K}^j_{\wedge_2}(t_2,\xi_2)$ satisfy the following
\begin{align}\label{P5eqn29}
\mathcal{K}^i_{\wedge_1}(t_1+rk_1,\xi_1)=\mathcal{K}^i_{\wedge_1}\left(t_1,\xi_1+\frac{2rk_1A_1}{B_1}\right)\phi^i_{\wedge_1,r}(k_1,\xi_1),
\end{align}
where
\begin{align}\label{P5eqn30}
\phi^i_{\wedge_1,r}(k_1,\xi_1)=e^{-i\left(A_1r^2k_1^2+D_1rk_1+B_1rk_1\xi_1-\frac{4r^2A_1^2C_1k_1^2}{B_1^2}-\frac{4rA_1C_1k_1\xi_1}{B_1}-\frac{2rA_1k_1}{B_1}\right)}
\end{align}
and 
\begin{align}\label{P5eqn31}
\mathcal{K}^j_{\wedge_2}(t_2+rk_2,\xi_2)=\mathcal{K}^j_{\wedge_2}\left(t_2,\xi_2+\frac{2rk_2A_2}{B_2}\right)\phi^j_{\wedge_2,r}(k_2,\xi_2),
\end{align}
where
\begin{align}\label{P5eqn32}
\phi^j_{\wedge_2,r}(k_2,\xi_2)=e^{-j\left(A_2r^2k_2^2+D_2rk_2+B_2rk_2\xi_2-\frac{4r^2A_2^2C_2k_2^2}{B_2^2}-\frac{4rA_2C_2k_2\xi_2}{B_2}-\frac{2rA_2k_2}{B_2}\right)}.
\end{align}
\end{lemma}
\begin{proof}
From the definition of $\mathcal{K}^i_{\wedge_1},$ we have
\begin{align*}
\mathcal{K}^i_{\wedge_1}(t_1+rk_1,\xi_1)
&=\frac{1}{\sqrt{2\pi}}e^{-i\left\{A_1(t_1+rk_1)^2+B_1(t_1+rk_1)\xi_1+C_1\xi_1^2+D_1(t_1+rk_1)+E_1\xi_1\right\}}\\
&=\frac{1}{\sqrt{2\pi}}e^{-i\left\{A_1t_1^2+B_1t_1\left(\xi_1+\frac{2rA_1k_1}{B_1}\right)+D_1t_1+C_1\xi_1^2+E_1\xi_1+B_1rk_1\xi_1\right\}}e^{-i(A_1r^2k_1^2+D_1rk_1)}\\
&=\frac{1}{\sqrt{2\pi}}e^{-i\left\{A_1t_1^2+B_1t_1\left(\xi_1+\frac{2rA_1k_1}{B_1}\right)+D_1t_1+C_1\left(\xi_1+\frac{2rA_1k_1}{B_1}\right)^2+E_1\left(\xi_1+\frac{2rA_1k_1}{B_1}\right)\right\}}\phi^i_{\wedge_1,r}(k_1,\xi_1),\\
\mbox{i.e.,}~\mathcal{K}^i_{\wedge_1}(t_1+rk_1,\xi_1)&=\mathcal{K}^i_{\wedge_1}\left(t_1,\xi_1+\frac{2rk_1A_1}{B_1}\right)\phi^i_{\wedge_1,r}(k_1,\xi_1).
\end{align*}
This proves equation \eqref{P5eqn29}. Similarly, equation \eqref{P5eqn31} can be proved.
\end{proof}
The theorem below gives the basic properties of the proposed STQQPFT.
\begin{theorem}\label{P5Theo4.2}
Let $g,g_1,g_2\in L^2_\mathbb{H}(\mathbb{R}^2)\cap L^\infty_\mathbb{H}(\mathbb{R}^2)$ be QWFs and  $f,f_1,f_2 \in L^2_\mathbb{H}(\mathbb{R}^2).$ Also let $\lambda\neq 0,~\boldsymbol k=(k_1,k_2)\in\mathbb{R}^2,$ $p,q\in\{x+iy:x,y\in\mathbb{R}\},~r,s\in\{x+jy:x,y\in\mathbb{R}\},$ then
\begin{enumerate}[label=(\roman*)]
\item{Boundedness}: $\left\|\mathcal{S}^{\wedge_1,\wedge_2}_{\mathbb{H},g}f\right\|_{L^\infty_\mathbb{H}(\mathbb{R}^2)}\leq\frac{1}{2\pi}\|g\|_{L^2_\mathbb{H}(\mathbb{R}^2)}\|f\|_{L^2_\mathbb{H}(\mathbb{R}^2)}.$
\item {Linearity}:
$\mathcal{S}^{\wedge_1,\wedge_2}_{\mathbb{H},g}(pf_1+qf_2)=p\left[\mathcal{S}^{\wedge_1,\wedge_2}_{\mathbb{H},g}f_1\right]+q\left[\mathcal{S}^{\wedge_1,\wedge_2}_{\mathbb{H},g}f_2\right]$
\item {Anti-linearity}:
$\mathcal{S}^{\wedge_1,\wedge_2}_{\mathbb{H},rg_1+sg_2}f=\left[\mathcal{S}^{\wedge_1,\wedge_2}_{\mathbb{H},g_1}f\right]\bar{r}+\left[\mathcal{S}^{\wedge_1,\wedge_2}_{\mathbb{H},g_2}f\right]\bar{s}.$
\item{Translation}: $\left(\mathcal{S}^{\wedge_1,\wedge_2}_{\mathbb{H},g}(\tau_{\boldsymbol k}f)\right)(\boldsymbol x,\boldsymbol \xi)=\phi^1_{\wedge_1}(k_1,\xi_1)\left(\mathcal{S}^{\wedge_1,\wedge_2}_{\mathbb{H},g}f\right)(\boldsymbol x-\boldsymbol k,\boldsymbol \xi'_{\boldsymbol x})\phi^j_{\wedge_2}(k_2,\xi_2),$ where 
$(\tau_{\boldsymbol k}f)(\boldsymbol t)=f(\boldsymbol t-\boldsymbol k),$ $\xi'_{\boldsymbol x}=\left(\xi_1+\frac{2A_1x_1}{B_1},\xi_2+\frac{2A_2x_2}{B_2}\right),$  
$\phi^i_{\wedge_1,1}(k_1,\xi_1,)$ and $\phi^j_{\wedge_2,1}(k_2,\xi_2),$ are obtained from \eqref{P5eqn30} and \eqref{P5eqn32} by replacing $r=1.$
\item {Scaling}: $\left(\mathcal{S}^{\wedge_1,\wedge_2}_{\mathbb{H},g_\lambda}f_{\lambda}\right)(\boldsymbol x,\boldsymbol\xi)=\left(\mathcal{S}^{\wedge_1',\wedge_2'}_{\mathbb{H},g}f\right)\left(\frac{1}{\lambda}\boldsymbol x,\boldsymbol\xi\right),$ where $(f_{\lambda})(\boldsymbol t)=\frac{1}{\lambda} f\left(\frac{1}{\lambda}\boldsymbol t\right),$ $\wedge_l'=\left(\lambda^2A_l,\lambda B_l,C_l,\lambda D_l,E_l\right),~l=1,2.$
\end{enumerate}
\end{theorem}
\begin{proof}
The proof of $(i)$ and $(ii)$ are straight forward so we omit their proof.

$(iii)$ We have from the definition \ref{P5Defn4.1}
\begin{align*}
\left(\mathcal{S}^{\wedge_1,\wedge_2}_{\mathbb{H},g}(\tau_{\boldsymbol k}f)\right)(\boldsymbol x,\boldsymbol\xi)=\int_{\mathbb{R}^2}\mathcal{K}^i_{\wedge_1}(t_1+k_1,\xi_1)f(\boldsymbol t)\overline{g(\boldsymbol t-(\boldsymbol x-\boldsymbol k))}\mathcal{K}^j_{\wedge_2}(t_2+k_1,\xi_2)d\boldsymbol t.
\end{align*}
Using lemma \eqref{P5Lemma4.1}, we get
\begin{align*}
\left(\mathcal{S}^{\wedge_1,\wedge_2}_{\mathbb{H},g}(\tau_{\boldsymbol k}f)\right)(\boldsymbol x,\boldsymbol\xi)
&=\int_{\mathbb{R}^2}\mathcal{K}^i_{\wedge_1}\left(t_1,\xi_1+\frac{2A_1k_1}{B_1}\right)\phi^i_{\wedge_1,1}(k_1,\xi_1)f(\boldsymbol t)\overline{g(\boldsymbol t-(\boldsymbol x-\boldsymbol k))}\mathcal{K}^j_{\wedge_2}\left(t_2,\frac{2A_2k_2}{B_2}\right)\phi^j_{\wedge_2,1}(k_2,\xi_2)d\boldsymbol t\\
&=\phi^i_{\wedge_1,1}(k_1,\xi_1)\left\{\int_{\mathbb{R}^2}\mathcal{K}^i_{\wedge_1}\left(t_1,\xi_1+\frac{2A_1k_1}{B_1}\right)f(\boldsymbol t)\overline{g(\boldsymbol t-(\boldsymbol x-\boldsymbol k))}\mathcal{K}^j_{\wedge_2}\left(t_2,\frac{2A_2k_2}{B_2}\right)d\boldsymbol t\right\}\phi^j_{\wedge_2,1}(k_2,\xi_2).
\end{align*}
Thus, we have
\begin{align*}
\left(\mathcal{S}^{\wedge_1,\wedge_2}_{\mathbb{H},g}(\tau_{\boldsymbol k}f)\right)(\boldsymbol x,\boldsymbol\xi)=\phi^i_{\wedge_1,1}(k_1,\xi_1)\left(\mathcal{S}^{\wedge_1,\wedge_2}_{\mathbb{H},g}f\right)(\boldsymbol x-\boldsymbol k,\boldsymbol\xi_{\boldsymbol x})\phi^j_{\wedge_2,1}(k_2,\xi_2).
\end{align*}
This proves $(iii).$

$(iv)$
We have
\begin{align}\label{P5eqn33}
\left(\mathcal{S}^{\wedge_1,\wedge_2}_{\mathbb{H},g_{\lambda}}f_{\lambda}\right)(\boldsymbol x,\boldsymbol\xi)
&=\int_{\mathbb{R}^2}\mathcal{K}^i_{\wedge_1}(\lambda t_1,\xi_1)f(\boldsymbol t)\overline{g\left(\boldsymbol t-\frac{1}{\lambda}\boldsymbol x\right)}\mathcal{K}^j_{\wedge_1}(\lambda t_2,\xi_2)d\boldsymbol t.
\end{align}
Now, 
\begin{align}\label{P5eqn34}
\mathcal{K}^i_{\wedge_1}(\lambda t_1,\xi_1)
&=\frac{1}{\sqrt{2\pi}}e^{-i\left((\lambda^2A_1)t_1^2+(\lambda B_1)t_1\xi_1+C_1\xi_1^2+D_1t_1+E_1\xi_1\right)}\notag\\
&=\mathcal{K}^i_{\wedge_1'}(t_1,\xi_1).
\end{align}
Similarly,
\begin{align}\label{P5eqn35}
\mathcal{K}^j_{\wedge_2}(\lambda t_2,\xi_2)=\mathcal{K}^j_{\wedge_2'}(t_2,\xi_2).
\end{align}
Using equations \eqref{P5eqn34} and \eqref{P5eqn35} in equation \eqref{P5eqn33}, we get
\begin{align*}
\left(\mathcal{S}^{\wedge_1,\wedge_2}_{\mathbb{H},g_{\lambda}}f_{\lambda}\right)(\boldsymbol x,\boldsymbol\xi)&=\int_{\mathbb{R}^2}\mathcal{K}^i_{\wedge_1'}(t_1,\xi_1)f(\boldsymbol t)\overline{g\left(\boldsymbol t-\frac{1}{\lambda}\boldsymbol x\right)}\mathcal{K}^j_{\wedge_2'}(t_2,\xi_2)d\boldsymbol t,\\
i.e., \left(\mathcal{S}^{\wedge_1,\wedge_2}_{\mathbb{H},g_{\lambda}}f_{\lambda}\right)(\boldsymbol x,\boldsymbol\xi)&=\left(\mathcal{S}^{\wedge_1',\wedge_2'}_{\mathbb{H},g}f\right)\left(\frac{1}{\lambda}\boldsymbol x,\boldsymbol\xi\right).
\end{align*}
This completes the proof.
\end{proof}
\begin{theorem}\label{P5Theo4.3}
(Inner product relation): If $g_1,g_2$ be two QWFs and $f_1,f_2\in L^2_\mathbb{H}(\mathbb{R}^2),$ then $\mathcal{S}^{\wedge_1,\wedge_2}_{\mathbb{H},g_1}f_1,~\mathcal{S}^{\wedge_1,\wedge_2}_{\mathbb{H},g_2}f_2\in L^2_\mathbb{H}(\mathbb{R}^2\times\mathbb{R}^2)$ and 
\begin{align}\label{P5eqn36}
\left\langle\mathcal{S}^{\wedge_1,\wedge_2}_{\mathbb{H},g_1}f_1,\mathcal{S}^{\wedge_1,\wedge_2}_{\mathbb{H},g_2}f_2\right\rangle=\frac{1}{|B_1B_2|}\langle f_1(\overline{g_1},\overline{g_2}),f_2\rangle.
\end{align}
\end{theorem}
\begin{proof}
We have
\begin{align*}
\int_{\mathbb{R}^2}\int_{\mathbb{R}^2}\left|\left(\mathcal{S}^{\wedge_1,\wedge_2}_{\mathbb{H},g}\right)(\boldsymbol x,\boldsymbol\xi)\right|^2d\boldsymbol x\boldsymbol\xi
&=\int_{\mathbb{R}^2}\left\{\int_{\mathbb{R}^2}\left|\left(\mathcal{Q}^{\wedge_1,\wedge_2}_{\mathbb{H}}\{f_1(\cdot)\overline{g_1(\cdot-\boldsymbol x)}\}\right)(\boldsymbol\xi)\right|^2d\boldsymbol\xi\right\}d\boldsymbol x\\
&=\frac{1}{|B_1B_2|}\int_{\mathbb{R}^2}\left\{\int_{\mathbb{R}^2}|f(\boldsymbol t)\overline{g(\boldsymbol t-\boldsymbol x)}|^2d\boldsymbol t\right\}d\boldsymbol x,~\mbox{using Parseval's Identity}\\
&=\frac{1}{|B_1B_2|}\|f\|^2_{L^2_\mathbb{H}(\mathbb{R}^2)}\|g\|^2_{L^2_\mathbb{H}(\mathbb{R}^2)}.
\end{align*}
Thus, $\mathcal{S}^{\wedge_1,\wedge_2}_{\mathbb{H},g_1}f_1\in L^2_\mathbb{H}(\mathbb{R}^2\times\mathbb{R}^2).$ Similarly, $\mathcal{S}^{\wedge_1,\wedge_2}_{\mathbb{H},g_2}f_2\in L^2_\mathbb{H}(\mathbb{R}^2\times\mathbb{R}^2).$

Now,
\begin{align*}
\left\langle\mathcal{S}^{\wedge_1,\wedge_2}_{\mathbb{H},g_1}f_1,\mathcal{S}^{\wedge_1,\wedge_2}_{\mathbb{H},g_2}f_2\right\rangle
&=Sc\int_{\mathbb{R}^2}\int_{\mathbb{R}^2}\left(\mathcal{Q}^{\wedge_1,\wedge_2}_{\mathbb{H}}\{f_1(\cdot)\overline{g_1(\cdot-\boldsymbol x)}\}\right)(\boldsymbol\xi)\overline{\left(\mathcal{Q}^{\wedge_1,\wedge_2}_{\mathbb{H}}\{f_2(\cdot)\overline{g_2(\cdot-\boldsymbol x)}\}\right)(\boldsymbol\xi)}d\boldsymbol xd\boldsymbol\xi\\
&=\frac{1}{|B_1B_2|}Sc\int_{\mathbb{R}^2}\left\{\int_{\mathbb{R}^2}f_1(\boldsymbol t)\overline{g_1(\boldsymbol t-\boldsymbol x)}~\overline{f_2(\boldsymbol t)\overline{g_{2}(\boldsymbol t-\boldsymbol x)}}d\boldsymbol t\right\}d\boldsymbol x\\
%&=\frac{1}{|B_1B_2|}Sc\int_{\mathbb{R}^2}f_1(\boldsymbol t)
%\left(\int_{\mathbb{R}^2}\overline{g_1(\boldsymbol t-\boldsymbol x)}%g_2(\boldsymbol t-\boldsymbol x)d\boldsymbol x\right)
%\overline{f_2(\boldsymbol t)}d\boldsymbol t\\
&=\frac{1}{|B_1B_2|}Sc\int_{\mathbb{R}^2}f_1(\boldsymbol t)\left(\overline{g_1},\overline{g_2}\right)\overline{f_2(\boldsymbol t)}d\boldsymbol t.
\end{align*}
Thus, it follows that
\begin{align*}
\left\langle\mathcal{S}^{\wedge_1,\wedge_2}_{\mathbb{H},g_1}f_1,\mathcal{S}^{\wedge_1,\wedge_2}_{\mathbb{H},g_2}f_2\right\rangle=\frac{1}{|B_1B_2|}\langle f_1\left(\overline{g_1},\overline{g_2}\right),f_2\rangle.
\end{align*}
This finishes the proof.
\end{proof}
\begin{remark}
From theorem \ref{P5Theo4.3}, we have the following results:
\begin{enumerate}
\item If $g_1=g_2=g$ in equation \eqref{P5eqn36}, then 
\begin{align*}
\left\langle\mathcal{S}^{\wedge_1,\wedge_2}_{\mathbb{H},g_1}f_1,\mathcal{S}^{\wedge_1,\wedge_2}_{\mathbb{H},g_2}f_2\right\rangle=\frac{1}{|B_1B_2|}\|g\|^2_{L^2_\mathbb{H}(\mathbb{R}^2)}\langle f_1,f_2\rangle.
\end{align*}
\item If $f_1=f_2=f$ in equation \eqref{P5eqn36}, then
\begin{align*}
\left\langle\mathcal{S}^{\wedge_1,\wedge_2}_{\mathbb{H},g_1}f_1,\mathcal{S}^{\wedge_1,\wedge_2}_{\mathbb{H},g_2}f_2\right\rangle=\frac{1}{|B_1B_2|}\|f\|^2_{L^2_\mathbb{H}(\mathbb{R}^2)}\langle g_1,g_2\rangle.
\end{align*}
\item If $f_1=f=f_2$ and $g_1=g=g_2$ in equation \eqref{P5eqn36}, then 
\begin{align}\label{P5eqn37}
\|\mathcal{S}^{\wedge_1,\wedge_2}_{\mathbb{H},g}f\|^2_{L^2_\mathbb{H}(\mathbb{R}^2\times\mathbb{R}^2)}=\frac{1}{|B_1B_2|}\|f\|^2_{L^2_\mathbb{H}(\mathbb{R}^2)}\|g\|^2_{L^2_\mathbb{H}(\mathbb{R}^2)}.
\end{align}
\end{enumerate} 
\end{remark}
The theorem below gives the reconstruction formula for the STQQPFT.
\begin{theorem}\label{P5Theo4.4}
(Inversion formula): Let $g$ be a QWF and $f\in L^2_\mathbb{H}(\mathbb{R}^2),$ then
$$f(\boldsymbol t)=\frac{|B_1B_2|}{\|g\|^2_{L^2_\mathbb{H}(\mathbb{R}^2)}}\int_{\mathbb{R}^2}\int_{\mathbb{R}^2}\overline{\mathcal{K}^i_{\wedge_1}(t_1,\xi_1)}\left(\mathcal{S}^{\wedge_1,\wedge_2}_{\mathbb{H},g}f\right)(\boldsymbol x,\boldsymbol\xi)\overline{\mathcal{K}^j_{\wedge_2}(t_2,\xi_2)}g(\boldsymbol t-\boldsymbol x)d\boldsymbol xd\boldsymbol\xi.$$
\end{theorem}
\begin{proof}
We have
\begin{align*}
\langle f,h\rangle
&=\frac{|B_1B_2|}{\|g\|^2_{L^2_\mathbb{H}(\mathbb{R}^2)}}Sc\int_{\mathbb{R}^2}\int_{\mathbb{R}^2}\left(\mathcal{S}^{\wedge_1,\wedge_2}_{\mathbb{H},g}f\right)(\boldsymbol x,\boldsymbol\xi)\overline{\left\{\int_{\mathbb{R}^2}\mathcal{K}^i_{\wedge_1}(t_1,\xi_1)h(\boldsymbol t)\overline{g(\boldsymbol t-\boldsymbol x)}\mathcal{K}^j_{\wedge_2}(t_2,\xi_2)d\boldsymbol t\right\}}d\boldsymbol xd\boldsymbol\xi\\
&=\frac{|B_1B_2|}{\|g\|^2_{L^2_\mathbb{H}(\mathbb{R}^2)}}\int_{\mathbb{R}^2}\int_{\mathbb{R}^2}\int_{\mathbb{R}^2}Sc\left\{\left(\mathcal{S}^{\wedge_1,\wedge_2}_{\mathbb{H},g}f\right)(\boldsymbol x,\boldsymbol\xi)\overline{\mathcal{K}^j_{\wedge_2}(t_2,\xi_2)}g(\boldsymbol t-\boldsymbol x)\overline{h(\boldsymbol t)}\overline{\mathcal{K}^i_{\wedge_1}(t_1,\xi_1)}\right\}d\boldsymbol td\boldsymbol xd\boldsymbol\xi\\
&=\frac{|B_1B_2|}{\|g\|^2_{L^2_\mathbb{H}(\mathbb{R}^2)}}Sc\int_{\mathbb{R}^2}\left\{\int_{\mathbb{R}^2}\int_{\mathbb{R}^2} \overline{\mathcal{K}^i_{\wedge_1}(t_1,\xi_1)}\left(\mathcal{S}^{\wedge_1,\wedge_2}_{\mathbb{H},g}f\right)(\boldsymbol x,\boldsymbol\xi)\overline{\mathcal{K}^j_{\wedge_2}(t_2,\xi_2)}g(\boldsymbol t-\boldsymbol x)d\boldsymbol xd\boldsymbol\xi\right\}\overline{h(\boldsymbol t)}d\boldsymbol t\\
&=\frac{|B_1B_2|}{\|g\|^2_{L^2_\mathbb{H}(\mathbb{R}^2)}}\left\langle\int_{\mathbb{R}^2}\int_{\mathbb{R}^2} \overline{\mathcal{K}^i_{\wedge_1}(t_1,\xi_1)}\left(\mathcal{S}^{\wedge_1,\wedge_2}_{\mathbb{H},g}f\right)(\boldsymbol x,\boldsymbol\xi)\overline{\mathcal{K}^j_{\wedge_2}(t_2,\xi_2)}g(\cdot-\boldsymbol x)d\boldsymbol xd\boldsymbol\xi,h(\cdot)\right\rangle.
\end{align*}
Since $h\in L^2_\mathbb{H}(\mathbb{R}^2)$ is arbitrary, it follows that
$$f(\boldsymbol t)=\frac{|B_1B_2|}{\|g\|^2_{L^2_\mathbb{H}(\mathbb{R}^2)}}\int_{\mathbb{R}^2}\int_{\mathbb{R}^2}\overline{\mathcal{K}^i_{\wedge_1}(t_1,\xi_1)}\left(\mathcal{S}^{\wedge_1,\wedge_2}_{\mathbb{H},g}f\right)(\boldsymbol x,\boldsymbol\xi)\overline{\mathcal{K}^j_{\wedge_2}(t_2,\xi_2)}g(\boldsymbol t-\boldsymbol x)d\boldsymbol xd\boldsymbol\xi.$$
This completes the proof.
\end{proof}

\subsection{Quaternion ambiguity function and Wigner-Ville distribution associated to the QQPFT}
In this subsection, we give the definitions of two sided QQPAF and QQPWVD and obtain their relation with that of the proposed STQQPFT.   
\begin{definition}\label{P5Defn4.2}
The two-sided quaternion quadratic phase ambiguity function (QQPAF) of  $f,g\in L^2_\mathbb{H}(\mathbb{R}^2),$ is defined by 
\begin{align*}
\left(\mathcal{A}^{\wedge_1,\wedge_2}_{\mathbb{H}}(f,g)\right)(\boldsymbol x,\boldsymbol\xi)=\int_{\mathbb{R}^2}\mathcal{K}^i_{\wedge_1}(t_1,\xi_1)f\left(\boldsymbol t+\frac{1}{2}\boldsymbol x\right)\overline{g\left(\boldsymbol t-\frac{1}{2}\boldsymbol x\right)}\mathcal{K}^j_{\wedge_2}\left(t_2,\xi_2\right)d\boldsymbol t, 
\end{align*}
where $\mathcal{K}^i_{\wedge_1}(t_1,\xi_1)$ and $\mathcal{K}^j_{\wedge_2}(t_2,\xi_2)$ are given by equations \eqref{P5eqn12} and \eqref{P5eqn13} respectively.
\end{definition}
The following theorem gives the relation between the QQPAF and the STQQPFT.
\begin{theorem}\label{P5Theo4.5}
If $g$ be a QWF and $f\in L^2_\mathbb{H}(\mathbb{R}^2),$ then 
\begin{align*}
\left(\mathcal{A}^{\wedge_1,\wedge_2}_{\mathbb{H}}(f,g)\right)(\boldsymbol x,\boldsymbol\xi)=\phi^i_{\wedge_1,-\frac{1}{2}}(x_1,\xi_1)\left(\mathcal{S}^{\wedge_1,\wedge_2}_{\mathbb{H},g}f\right)(\boldsymbol x,\boldsymbol\xi'_{\boldsymbol x})\phi^j_{\wedge_2,-\frac{1}{2}}(x_2,\xi_2),~\boldsymbol\xi'_{\boldsymbol x}=\left(\xi_1-\frac{A_1x_1}{B_1},\xi_2-\frac{A_2x_2}{B_2}\right)
\end{align*}
where $\phi^i_{\wedge_1,-\frac{1}{2}}(x_1,\xi_1)$ and $\phi^j_{\wedge_2,-\frac{1}{2}}(x_2,\xi_2)$ are obtained from equations \eqref{P5eqn30} and \eqref{P5eqn32} by replacing $r=-\frac{1}{2}$.
\end{theorem} 
\begin{proof}
From the definition of $\mathcal{A}^{\wedge_1,\wedge_2}_{\mathbb{H}}(f,g),$ it follows that
\begin{align*}
\left(\mathcal{A}^{\wedge_1,\wedge_2}_{\mathbb{H}}(f,g)\right)(\boldsymbol x,\boldsymbol\xi)=\int_{\mathbb{R}^2}\mathcal{K}^i_{\wedge_1}\left(t_1-\frac{x_1}{2},\xi_2\right)f(\boldsymbol t)\overline{g(\boldsymbol t-\boldsymbol x)}\mathcal{K}^j_{\wedge_2}\left(t_2-\frac{x_2}{2},\xi_2\right)d\boldsymbol t.
\end{align*}
Using equations \eqref{P5eqn29} and \eqref{P5eqn31} for $r=-\frac{1}{2},$ we get
\begin{align*}
\left(\mathcal{A}^{\wedge_1,\wedge_2}_{\mathbb{H}}(f,g)\right)(\boldsymbol x,\boldsymbol\xi)
&=\int_{\mathbb{R}^2}\mathcal{K}^i_{\wedge_1}\left(t_1,\xi_1-\frac{A_1x_1}{B_1}\right)\phi^i_{\wedge_1,-\frac{1}{2}}(x_1,\xi_1)f(\mathbb{\boldsymbol t})\overline{g(\boldsymbol t-\boldsymbol x)}\mathcal{K}^j_{\wedge_2}\left(t_2,\xi_2-\frac{A_2x_2}{B_2}\right)\phi^j_{\wedge_2,-\frac{1}{2}}(x_2,\xi_2)d\boldsymbol t\\
&=\phi^i_{\wedge_1,-\frac{1}{2}}(x_1,\xi_1)\left\{\int_{\mathbb{R}^2}\mathcal{K}^i_{\wedge_1}\left(t_1,\xi_1-\frac{A_1x_1}{B_1}\right)f(\mathbb{\boldsymbol t})\overline{g(\boldsymbol t-\boldsymbol x)}\mathcal{K}^j_{\wedge_2}\left(t_2,\xi_2-\frac{A_2x_2}{B_2}\right)d\boldsymbol t\right\}\phi^j_{\wedge_2,-\frac{1}{2}}(x_2,\xi_2).
\end{align*}
This gives
\begin{align*}
\left(\mathcal{A}^{\wedge_1,\wedge_2}_{\mathbb{H}}(f,g)\right)(\boldsymbol x,\boldsymbol\xi)=\phi^i_{\wedge_1,-\frac{1}{2}}(x_1,\xi_1)\left(\mathcal{S}^{\wedge_1,\wedge_2}_{\mathbb{H},g}f\right)(\boldsymbol x,\boldsymbol\xi'_{\boldsymbol x})\phi^j_{\wedge_2,-\frac{1}{2}}(x_2,\xi_2).
\end{align*}
This completes the proof.

\end{proof}
\begin{definition}\label{P5Defn4.3}
The two-sided quaternion quadratic phase Wigner-Ville distribution (QQPWVD) of  $f,g\in L^2_\mathbb{H}(\mathbb{R}^2),$ is defined by
\begin{align*}
\left(\mathcal{W}^{\wedge_1,\wedge_2}_{\mathbb{H}}(f,g)\right)(\boldsymbol x,\boldsymbol\xi)=\int_{\mathbb{R}^2}\mathcal{K}^i_{\wedge_1}(t_1,\xi_1)f\left(\boldsymbol x+\frac{1}{2}\boldsymbol t\right)\overline{g\left(\boldsymbol x-\frac{1}{2}\boldsymbol t\right)}\mathcal{K}^j_{\wedge_2}\left(t_2,\xi_2\right)d\boldsymbol t, 
\end{align*}
where $\mathcal{K}^i_{\wedge_1}(t_1,\xi_1)$ and $\mathcal{K}^j_{\wedge_2}(t_2,\xi_2)$ are given by equations \eqref{P5eqn12} and \eqref{P5eqn13} respectively.
\end{definition}
The following theorem gives the relation between the QQPWVD and the STQQPFT.
\begin{theorem}\label{P5Theo4.6}
If $g$ be a QWF and $f\in L^2_\mathbb{H}(\mathbb{R}^2),$ then 
\begin{align*}
\left(\mathcal{W}^{\wedge_1,\wedge_2}_{\mathbb{H}}(f,g)\right)(\boldsymbol x,\boldsymbol\xi)=4\psi^i_{\wedge_1}(x_1,\xi_1)\left(\mathcal{S}^{\wedge_1',\wedge_2'}_{\mathbb{H},\tilde{g}}f\right)(2\boldsymbol x,\boldsymbol\xi'_{\boldsymbol x})\psi^j_{\wedge_2}(x_2,\xi_2),~\boldsymbol\xi'_{\boldsymbol x}=\left(\xi_1-\frac{4A_1x_1}{B_1},\xi_2-\frac{4A_2x_2}{B_2}\right)
\end{align*}
where $\wedge_l'=(4A_l,2B_l,C_l,2D_l,E_l),~l=1,2,$ $\tilde{g}(\boldsymbol t)=g(-\boldsymbol t),$ 
$$\psi^i_{\wedge_1}(x_1,\xi_1)=e^{-i\left(4A_1x_1^2-2B_1x_1\xi_1-2D_1x_1-\frac{16A_1^2C_1x_1^2}{B_1^2}+\frac{8A_1C_1x_1\xi_1}{B_1}+\frac{4A_1E_1x_1}{B_1}\right)}$$ 
and $$\psi^j_{\wedge_2}(x_2,\xi_2)=e^{-j\left(4A_2x_2^2-2B_2x_2\xi_2-2D_2x_2-\frac{16A_2^2C_2x_2^2}{B_2^2}+\frac{8A_2C_2x_2\xi_2}{B_2}+\frac{4A_2E_2x_2}{B_2}\right)}.$$
\end{theorem}
\begin{proof}
From the definition of $\mathcal{W}^{\wedge_1,\wedge_2}_{\mathbb{H}}(f,g),$ we have
\begin{align}\label{P5eqn38}
\left(\mathcal{W}^{\wedge_1,\wedge_2}_{\mathbb{H}}(f,g)\right)(\boldsymbol x,\boldsymbol\xi)
&=4\int_{\mathbb{R}^2}\mathcal{K}^i_{\wedge_1}(2(t_1-x_1),\xi_1)f(\boldsymbol t)\overline{g(2\boldsymbol x-\boldsymbol t)}\mathcal{K}^j_{\wedge_2}(2(t_2-x_2),\xi_2)d\boldsymbol t\notag\\
&=4\int_{\mathbb{R}^2}\mathcal{K}^i_{\wedge_1}(2(t_1-x_1),\xi_1)f(\boldsymbol t)\overline{\tilde{g}(\boldsymbol t-2\boldsymbol x)}\mathcal{K}^j_{\wedge_2}(2(t_2-x_2),\xi_2)d\boldsymbol t.
\end{align}
Now from the definition of $\mathcal{K}^i_{\wedge_1},$ in equation \eqref{P5eqn12}, we have
\begin{align*}
\mathcal{K}^i_{\wedge_1}(2(t_1-x_1),\xi_1)
&=\frac{1}{\sqrt{2\pi}}e^{-i(4A_1t_1^2-8A_1x_1t_1+2B_1t_1\xi_1+2D_1t_1+E_1\xi_1+C_1\xi_2)}e^{-i(4A_1x_1^2-2B_1x_1\xi_1-2D_1x_1)}\\
&=\frac{1}{\sqrt{2\pi}}e^{-i\left\{(4A_1)t_1^2+2B_1\left(\xi_1-\frac{4A_1x_1}{B_1}\right)+C_1\left(\xi_1-\frac{4A_1x_1}{B_1}\right)^2+(2D_1)t_1+E_1\left(\xi_1-\frac{4A_1x_1}{B_1}\right)\left(\xi_1-\frac{4A_1x_1}{B_1}\right)\right\}}\psi^i_{\wedge_1}(x_1,\xi_1)
\end{align*}
i.e., 
\begin{align}\label{P5eqn39}
\mathcal{K}^i_{\wedge_1}(2(t_1-x_1),\xi_1)=\mathcal{K}^i_{\wedge_1'}\left(t_1,\xi_1-\frac{4A_1x_1}{B_1}\right)\psi^i_{\wedge_1}(x_1,\xi_1).
\end{align}
Similarly, we have
\begin{align}\label{P5eqn40}
\mathcal{K}^j_{\wedge_2}(2(t_2-x_2),\xi_2)=\mathcal{K}^j_{\wedge_2'}\left(t_2,\xi_2-\frac{4A_2x_2}{B_2}\right)\psi^j_{\wedge_2}(x_2,\xi_2).
\end{align}
Using equations \eqref{P5eqn39} and \eqref{P5eqn40} in \eqref{P5eqn38}, we have
\begin{align*}
\left(\mathcal{W}^{\wedge_1,\wedge_2}_{\mathbb{H}}(f,g)\right)(\boldsymbol x,\boldsymbol\xi)
&=4\int_{\mathbb{R}^2}\mathcal{K}^i_{\wedge_1'}\left(t_1,\xi_1-\frac{4A_1x_1}{B_1}\right)\psi^i_{\wedge_1}(x_1,\xi_1)f(\boldsymbol t)\overline{\tilde{g}(\boldsymbol t-2\boldsymbol x)}\mathcal{K}^j_{\wedge_2'}\left(t_2,\xi_2-\frac{4A_2x_2}{B_2}\right)\psi^j_{\wedge_2}(x_2,\xi_2)d\boldsymbol t\\
&=4\psi^i_{\wedge_1}(x_1,\xi_1)\left\{\int_{\mathbb{R}^2}\mathcal{K}^i_{\wedge_1'}\left(t_1,\xi_1-\frac{4A_1x_1}{B_1}\right)f(\boldsymbol t)\overline{\tilde{g}(\boldsymbol t-2\boldsymbol x)}\mathcal{K}^j_{\wedge_2'}\left(t_2,\xi_2-\frac{4A_2x_2}{B_2}\right)d\boldsymbol t\right\}\psi^j_{\wedge_2}(x_2,\xi_2).
\end{align*}
This gives
\begin{align*}
\left(\mathcal{W}^{\wedge_1,\wedge_2}_{\mathbb{H}}(f,g)\right)(\boldsymbol x,\boldsymbol\xi)=4\psi^i_{\wedge_1}(x_1,\xi_1)\left(\mathcal{S}^{\wedge_1',\wedge_2'}_{\mathbb{H},\tilde{g}}f\right)(2\boldsymbol x,\boldsymbol\xi'_{\boldsymbol x})\psi^j_{\wedge_2}(x_2,\xi_2).
\end{align*}
This completes the proof.
\end{proof}

\subsection{Uncertainty principle for STQQPFT}
The Heisenberg's UP gives the information about a function and its FT, it says that the function cannot be highly localized in both time and frequency domain. Wilczok (\cite{wilczok2000new}) introduced a new class of UP that compares the localization of a functions with the localization of its wavelet transform, analogous to the Heisenberg UP governing the localization of the complex valued function and the corresponding FT. Gupta et al. \cite{gupta2022linear} obtained the Lieb's and Donoho-Stark's UP for the linear canonical wavelet transform and obtained the lower bound of the measure of its essential support.

Here, we prove the Lieb's UP for the STQQPFT, QQPWVD and QQPAF. Analogous result for the classical STFT and the windowed linear canonical transform can be found in \cite{grochenig2001foundations} and  \cite{kou2012paley} respectively. Before we move forward, let us first prove the following lemma. 

\begin{lemma}\label{P5Lemma4.7}
(Lieb's inequality)
Let $g$ be a QWF, $f\in L^2_\mathbb{H}(\mathbb{R}^2)$ and $2\leq q<\infty.$ Then 
\begin{align}\label{P5eqn41}
\left\|\mathcal{S}^{\wedge_1,\wedge_2}_{\mathbb{H},g}f \right\|_{L^q_\mathbb{H}(\mathbb{R}^2\times\mathbb{R}^2)}\leq\frac{(2\pi)^{\frac{1}{q}-\frac{1}{p}}}{|B_1B_2|^{\frac{1}{q}}}\left(\frac{2}{q}\right)^\frac{2}{q}\|g\|_{L^2_\mathbb{H}(\mathbb{R}^2)}\|f\|_{L^2_\mathbb{H}(\mathbb{R}^2)}.
\end{align}
\end{lemma}
\begin{proof}
\begin{align}\label{P5eqn42}
\left(\int_{\mathbb{R}^2}\left|\left(\mathcal{S}^{\wedge_1,\wedge_2}_{\mathbb{H},g}f \right)(\boldsymbol x,\boldsymbol\xi)\right|^qd\boldsymbol\xi\right)^\frac{1}{q}=\left(\int_{\mathbb{R}^2}\left|\left(\mathcal{Q}^{\wedge_1,\wedge_2}_{\mathbb{H}}\{f(\cdot)\overline{g(\cdot-\boldsymbol x)}\}\right)(\boldsymbol\xi)\right|^qd\boldsymbol\xi\right)^\frac{1}{q}.
\end{align}
Using Hausdorff-Young inequality, we get
\begin{align*}
\left(\int_{\mathbb{R}^2}\left|\left(\mathcal{S}^{\wedge_1,\wedge_2}_{\mathbb{H},g}f \right)(\boldsymbol x,\boldsymbol\xi)\right|^qd\boldsymbol\xi\right)^\frac{1}{q}
&\leq\frac{A^2_p(2\pi)^{\frac{1}{q}-\frac{1}{p}}}{|B_1B_2|^{\frac{1}{q}}}\left(\int_{\mathbb{R}^2}\left|f(\boldsymbol t)\overline{g(\boldsymbol t-\boldsymbol x)}\right|^pd\boldsymbol t\right)^{\frac{1}{p}}\\
&=\frac{A^2_p(2\pi)^{\frac{1}{q}-\frac{1}{p}}}{|B_1B_2|^{\frac{1}{q}}}\left(\int_{\mathbb{R}^2}|f(\boldsymbol t)|^p|\tilde{g}(\boldsymbol x-\boldsymbol t)|^pd\boldsymbol t\right)^{\frac{1}{p}},~\tilde{g}(\boldsymbol t)=g(-\boldsymbol t)\\
&=\frac{A^2_p(2\pi)^{\frac{1}{q}-\frac{1}{p}}}{|B_1B_2|^{\frac{1}{q}}}\left\{\left(|f|^p\star|\tilde{g}|^p\right)(\boldsymbol x)\right\}^{\frac{1}{p}}.
\end{align*}
This implies that
\begin{align*}
\int_{\mathbb{R}^2}\int_{\mathbb{R}^2}\left|\left(\mathcal{S}^{\wedge_1,\wedge_2}_{\mathbb{H},g}f \right)(\boldsymbol x,\boldsymbol\xi)\right|^qd\boldsymbol xd\boldsymbol\xi\leq \frac{A^{2q}_p(2\pi)^{q(\frac{1}{q}-\frac{1}{p})}}{|B_1B_2|}\int_{\mathbb{R}^2}\left\{\left(|f|^p\star|\tilde{g}|^p\right)(\boldsymbol x)\right\}^{\frac{q}{p}}d\boldsymbol x.
\end{align*}
This gives
\begin{align}\label{P5eqn43}
\left\{\int_{\mathbb{R}^2}\int_{\mathbb{R}^2}\left|\left(\mathcal{S}^{\wedge_1,\wedge_2}_{\mathbb{H},g}f \right)(\boldsymbol x,\boldsymbol\xi)\right|^qd\boldsymbol xd\boldsymbol\xi\right\}^\frac{1}{q}
&\leq \frac{A^2_p(2\pi)^{\frac{1}{q}-\frac{1}{p}}}{|B_1B_2|^{\frac{1}{q}}}\left[\int_{\mathbb{R}^2}\left\{\left(|f|^p\star|\tilde{g}|^p\right)(\boldsymbol x)\right\}^{\frac{q}{p}}d\boldsymbol x\right]^{\frac{q}{p}\cdot\frac{1}{q}}\notag\\
&=\frac{A^2_p(2\pi)^{\frac{1}{q}-\frac{1}{p}}}{|B_1B_2|^{\frac{1}{q}}}\left\||f|^p\star|\tilde{g}|^p\right\|^{\frac{1}{p}}_{L^{\frac{q}{p}}_\mathbb{H}(\mathbb{R}^2)}.
\end{align}
Now we see that, if $k=\frac{2}{p},~ l=\frac{q}{p},$ then $k\geq 1$ and $\frac{1}{k}+\frac{1}{k}=1+\frac{1}{l}.$
Since $|f|^p,~|\tilde{g}|^p\in L^k_\mathbb{H}(\mathbb{R}^2),$ we get, by Young's inequality
\begin{align}\label{P5eqn44}
\left\||f|^p\star|\tilde{g}|^p\right\|^{\frac{1}{p}}_{L^{\frac{q}{p}}_\mathbb{H}(\mathbb{R}^2)}\leq A_k^4A_{l'}^2\|f\|^p_{L^2_\mathbb{H}(\mathbb{R}^2)}\|\tilde{g}\|^p_{L^2_\mathbb{H}(\mathbb{R}^2)}.
\end{align}
Therefore, from equations \eqref{P5eqn43} and \eqref{P5eqn44}, it follows that
\begin{align}\label{P5eqn45}
\left\{\int_{\mathbb{R}^2}\int_{\mathbb{R}^2}\left|\left(\mathcal{S}^{\wedge_1,\wedge_2}_{\mathbb{H},g}f \right)(\boldsymbol x,\boldsymbol\xi)\right|^qd\boldsymbol xd\boldsymbol\xi\right\}^\frac{1}{q}\leq\frac{(2\pi)^{\frac{1}{q}-\frac{1}{p}}}{|B_1B_2|^{\frac{1}{q}}} A_p^2A_k^{\frac{4}{p}}A_{l'}^{\frac{2}{p}}\|g\|_{L^2_\mathbb{H}(\mathbb{R}^2)}\|f\|_{L^2_\mathbb{H}(\mathbb{R}^2)},
\end{align}
where $A_r=\left(\frac{r^{\frac{1}{r}}}{r'^{\frac{1}{r'}}}\right)^{\frac{1}{2}},~\frac{1}{r}+\frac{1}{r'}=1.$ 
Now, we have
\begin{align}\label{P5eqn46}
A_p^2A_k^{\frac{4}{p}}A_{l'}^{\frac{2}{p}}
&=\frac{p^{\frac{1}{p}}}{q\frac{1}{q}}\cdot\frac{k}{{k'}^\frac{2}{k'p}}\cdot\frac{{l'}^\frac{1}{pl'}}{\left(\frac{q}{p}\right)^{\frac{1}{q}}},~\mbox{since}~k=\frac{2}{q},~l=\frac{q}{p}\notag\\
&=\frac{p}{q^{\frac{2}{q}}}\cdot\frac{{l'}^{\frac{1}{pl'}}}{{k'}^{\frac{2}{k'p}}}\notag\\
&\leq \frac{2}{q^{\frac{2}{q}}}\cdot\left(\frac{1}{2}\right)^{\frac{q-p}{pq}},~\mbox{since}~k'=2l'\notag\\
&=\left(\frac{2}{q}\right)^{\frac{2}{q}}.
\end{align}
Thus using equation \eqref{P5eqn46} in \eqref{P5eqn45}, we get
\begin{align*}
\left\{\int_{\mathbb{R}^2}\int_{\mathbb{R}^2}\left|\left(\mathcal{S}^{\wedge_1,\wedge_2}_{\mathbb{H},g}f \right)(\boldsymbol x,\boldsymbol\xi)\right|^qd\boldsymbol xd\boldsymbol\xi\right\}^\frac{1}{q}\leq\frac{(2\pi)^{\frac{1}{q}-\frac{1}{p}}}{|B_1B_2|^{\frac{1}{q}}} \left(\frac{2}{q}\right)^{\frac{2}{q}}\|g\|_{L^2_\mathbb{H}(\mathbb{R}^2)}\|f\|_{L^2_\mathbb{H}(\mathbb{R}^2)}.
\end{align*}
This finishes the proof.
\end{proof}

\subsection{Lieb's uncertainty principle}
\begin{definition}\label{P5Defn4.4}
Let $\epsilon\geq 0$ and $\Omega\subset\mathbb{R}^n$ be measurable. A function $F\in L^2_\mathbb{H}(\mathbb{R}^n)$ is said to be $\epsilon-$concentrated on $\Omega$ if 
$$\|\chi_{\Omega^c}F\|_{L^2_\mathbb{H}(\mathbb{R}^n)}\leq \epsilon\|F\|_{L^2_\mathbb{H}(\mathbb{R}^n)},$$
where $\chi_{\Omega}$ denotes the indicator function on $\Omega.$
 
If $0\leq \epsilon\leq\frac{1}{2},$ then majority of the energy is concentrated on $\Omega$ and $\Omega$ is said to be the essential support of $F.$ Support of $F$ is contained in $\Omega,$ if $\epsilon=0.$ 
\end{definition}
\begin{theorem}\label{P5Theo4.8}
Let $g$ be a QWF and $f\in L^2_\mathbb{H}(\mathbb{R}^2),$ such that $f\neq 0.$ Let $\epsilon\geq 0$ and $\Omega\subset\mathbb{R}^2\times\mathbb{R}^2$ is a measurable set. If $\mathcal{S}^{\wedge_1,\wedge_2}_{\mathbb{H},g}f,$ on $\Omega,$ is $\epsilon-$concentrated, then for every $q>2$
\begin{align*}
|\Omega|\geq \frac{(2\pi)^2}{|B_1B_2|}(1-\epsilon^2)^{\frac{q}{q-2}}\left(\frac{q}{2}\right)^{\frac{4}{q-2}}.
\end{align*}
\end{theorem}
\begin{proof}
Since $\mathcal{S}^{\wedge_1,\wedge_2}_{\mathbb{H},g}f$ is $\epsilon-$concentrated on $\Omega,$ we have
\begin{align*}
\left\|\chi_{\Omega^c}\mathcal{S}^{\wedge_1,\wedge_2}_{\mathbb{H},g}f\right\|_{L^2_\mathbb{H}(\mathbb{R}^2\times\mathbb{R}^2)}\leq \frac{\epsilon^2}{|B_1B_2|}\|f\|^2_{L^2_\mathbb{H}(\mathbb{R}^2)}\|g\|^2_{L^2_\mathbb{H}(\mathbb{R}^2)}.
\end{align*}
This implies
\begin{align}\label{P5eqn47}
\left\|\chi_{\Omega}\mathcal{S}^{\wedge_1,\wedge_2}_{\mathbb{H},g}f\right\|_{L^2_\mathbb{H}(\mathbb{R}^2\times\mathbb{R}^2)}\geq \frac{1}{|B_1B_2|}(1-\epsilon^2)\|f\|^2_{L^2_\mathbb{H}(\mathbb{R}^2)}\|g\|^2_{L^2_\mathbb{H}(\mathbb{R}^2)}.
\end{align}
Now, using Holder's inequality, we have
\begin{align*}
\left\|\chi_{\Omega}\mathcal{S}^{\wedge_1,\wedge_2}_{\mathbb{H},g}f\right\|_{L^2_\mathbb{H}(\mathbb{R}^2\times\mathbb{R}^2)}\leq &\left\{\int_{\mathbb{R}^2}\int_{\mathbb{R}^2}\left(\chi_{\Omega}(\boldsymbol x,\boldsymbol\xi)\right)^{\frac{q}{q-2}}d\boldsymbol xd\boldsymbol\xi\right\}^{\frac{q}{2}}\left\{\int_{\mathbb{R}^2}\int_{\mathbb{R}^2}\left(\left|\left(\mathcal{S}^{\wedge_1,\wedge_2}_{\mathbb{H},g}f\right)(\boldsymbol x,\boldsymbol\xi)\right|^2\right)^{\frac{q}{q-2}}d\boldsymbol xd\boldsymbol\xi\right\}^{\frac{2}{q}}\\
&=|\Omega|^{\frac{q-2}{q}}\left\|\mathcal{S}^{\wedge_1,\wedge_2}_{\mathbb{H},g}f\right\|^2_{L^2_\mathbb{H}(\mathbb{R}^2)}.
\end{align*}
Using, the Lieb's inequality \eqref{P5eqn41}, we get
\begin{align}\label{P5eqn48}
\left\|\chi_{\Omega}\mathcal{S}^{\wedge_1,\wedge_2}_{\mathbb{H},g}f\right\|_{L^2_\mathbb{H}(\mathbb{R}^2\times\mathbb{R}^2)}\leq |\Omega|^{\frac{q-2}{q}}\frac{(2\pi)^{\frac{2}{q}-\frac{2}{p}}}{|B_1B_2|^{\frac{2}{q}}}\left(\frac{2}{q}\right)^{\frac{2}{q}}\|f\|^2_{L^2_\mathbb{H}(\mathbb{R}^2)}\|g\|^2_{L^2_\mathbb{H}(\mathbb{R}^2)}.
\end{align}
From equation \eqref{P5eqn47} and equation \eqref{P5eqn48}, we get
\begin{align*}
|\Omega|^{\frac{q-2}{q}}\frac{(2\pi)^{\frac{2}{q}-\frac{2}{p}}}{|B_1B_2|^{\frac{2}{q}}}\left(\frac{2}{q}\right)^{\frac{2}{q}}\geq \frac{1}{|B_1B_2|}(1-\epsilon^2).
\end{align*}
This gives
\begin{align*}
&|\Omega|\geq \frac{1}{|B_1B_2|}(2\pi)^{2(1-\frac{2}{q})\frac{q}{q-2}}(1-\epsilon^2)^{\frac{q}{q-2}}\left(\frac{q}{2}\right)^{\frac{4}{q-2}},~\mbox{since}~\frac{1}{p}+\frac{1}{q}=1\\
&\mbox{i.e.,}~|\Omega|\geq \frac{1}{|B_1B_2|}(2\pi)^{2}(1-\epsilon^2)^{\frac{q}{q-2}}\left(\frac{q}{2}\right)^{\frac{4}{q-2}}.
\end{align*}
This completes the proof.
\end{proof}
\begin{remark}
Taking $\epsilon=0,$ in the above theorem, we get the following lower bound for the support of $\mathcal{S}^{\wedge_1,\wedge_2}_{\mathbb{H},g}f$
\begin{align}
&\left|\rm{supp}\left(\mathcal{S}^{\wedge_1,\wedge_2}_{\mathbb{H},g}f\right)\right|\geq \frac{(2\pi)^2}{|B_1B_2|}\lim_{q\rightarrow 2+}\left(\frac{q}{2}\right)^{\frac{4}{q-2}}\notag\\
&\mbox{i.e.,}~\left|\rm{supp}\left(\mathcal{S}^{\wedge_1,\wedge_2}_{\mathbb{H},g}f\right)\right|\geq \frac{(2\pi e)^2}{|B_1B_2|}.
\end{align}
\end{remark}
i.e., measure of the support of $\mathcal{S}^{\wedge_1,\wedge_2}_{\mathbb{H},g}f\geq \frac{(2\pi e)^2}{|B_1B_2|}.$
\begin{corollary}
Let $g$ be a QWF and $f\in L^2_\mathbb{H}(\mathbb{R}^2),$ such that $f\neq 0.$ Let $\epsilon\geq 0$ and $\Omega\subset\mathbb{R}^2\times\mathbb{R}^2$ is measurable. If $\mathcal{A}^{\wedge_1,\wedge_2}_{\mathbb{H}}(f,g),$ on $\Omega,$ is $\epsilon-$concentrated, then for every $q>2$
\begin{align}\label{P5eqn50}
|\Omega|\geq \frac{(2\pi)^2}{|B_1B_2|}(1-\epsilon^2)^{\frac{q}{q-2}}\left(\frac{q}{2}\right)^{\frac{4}{q-2}}.
\end{align}
In particular, if $\epsilon=0,$ then 
\begin{align}\label{P5eqn51}
\left|\rm{supp}\left(\mathcal{A}^{\wedge_1,\wedge_2}_{\mathbb{H}}(f,g)\right)\right|\geq \frac{(2\pi e)^2}{|B_1B_2|}.
\end{align}
\end{corollary}
\begin{proof}
From theorem \ref{P5Theo4.5}, it follows that
$$\left|\left(\mathcal{A}^{\wedge_1,\wedge_2}_{\mathbb{H}}(f,g)\right)(\boldsymbol x,\boldsymbol\xi)\right|=\left|\left(\mathcal{S}^{\wedge_1,\wedge_2}_{\mathbb{H},g}f \right)(\boldsymbol x,\boldsymbol\xi'_{\boldsymbol x})\right|,~\boldsymbol\xi'_{\boldsymbol x}=\left(\xi_1-\frac{A_1x_1}{B_1},\xi_2-\frac{A_2x_2}{B_2}\right).$$
Since $\mathcal{A}^{\wedge_1,\wedge_2}_{\mathbb{H}}(f,g)$ is $\epsilon-$concentrated on $\Omega,$ it can be shown that $\mathcal{S}^{\wedge_1,\wedge_2}_{\mathbb{H},g}f$ is $\epsilon-$concentrated on $P^{-1}\Omega,$ where $P$ is the non-singular matrix given by
$
\begin{bmatrix}
1 & 0 & 0 & 0\\
0 & 1 & 0 & 0\\
\frac{A_1}{B_1} & 0 & 1 & 0\\
0 & \frac{A_2}{B_2} & 0 & 1
\end{bmatrix}
$
and $P^{-1}\Omega=\{P^{-1}\boldsymbol x:\boldsymbol x\in\Omega\}.$ So, by theorem \ref{P5Theo4.8}, we have
\begin{align*}
|P^{-1}\Omega|\geq \frac{(2\pi)^2}{|B_1B_2|}(1-\epsilon^2)^{\frac{q}{q-2}}\left(\frac{q}{2}\right)^{\left(\frac{4}{q-2}\right)}.
\end{align*}
This gives
\begin{align*}
|\Omega|\geq \frac{(2\pi)^2}{|B_1B_2|}(1-\epsilon^2)^{\frac{q}{q-2}}\left(\frac{q}{2}\right)^{\left(\frac{4}{q-2}\right)},~\mbox{since}~det(P^{-1})=1.
\end{align*}
This proves equation \eqref{P5eqn50}.
\end{proof}

\begin{corollary}
Let $g$ be a QWF and $f\in L^2_\mathbb{H}(\mathbb{R}^2),$ such that $f\neq 0.$ Let  $\epsilon\geq 0$ and $\Omega\subset\mathbb{R}^2\times\mathbb{R}^2$ is measurable. If $\mathcal{W}^{\wedge_1,\wedge_2}_{\mathbb{H}}(f,g),$ on $\Omega,$ is $\epsilon-$concentrated, then for every $q>2$
\begin{align}\label{P5eqn52}
|\Omega|\geq \frac{(2\pi)^2}{16|B_1B_2|}(1-\epsilon^2)^{\frac{q}{q-2}}\left(\frac{q}{2}\right)^{\frac{4}{q-2}}.
\end{align}
In particular, if $\epsilon=0,$ then 
\begin{align}\label{P5eqn53}
\left|\rm{supp}\left(\mathcal{W}^{\wedge_1,\wedge_2}_{\mathbb{H}}(f,g)\right)\right|\geq \frac{(\pi e)^2}{4|B_1B_2|}.
\end{align}
\end{corollary}
\begin{proof}
From theorem \ref{P5Theo4.6}, it follows that
$$\left|\left(\mathcal{W}^{\wedge_1,\wedge_2}_{\mathbb{H}}(f,g)\right)(\boldsymbol x,\boldsymbol\xi)\right|=4\left|\left(\mathcal{S}^{\wedge_1',\wedge_2'}_{\mathbb{H},\tilde{g}}f \right)(2\boldsymbol x,\boldsymbol\xi'_{\boldsymbol x})\right|,~\boldsymbol\xi'_{\boldsymbol x}=\left(\xi_1-\frac{4A_1x_1}{B_1},\xi_2-\frac{4A_2x_2}{B_2}\right).$$ 
Since $\mathcal{W}^{\wedge_1,\wedge_2}_{\mathbb{H}}(f,g)$ is $\epsilon-$concentrated on $\Omega,$ it can be shown that $\mathcal{S}^{\wedge_1',\wedge_2'}_{\mathbb{H},g}f$ is $\epsilon-$concentrated on $P^{-1}\Omega,$ where $P$ is the non-singular matrix given by
$
\begin{bmatrix}
\frac{1}{2} & 0 & 0 & 0\\
0 & \frac{1}{2} & 0 & 0\\
\frac{4A_1}{B_1} & 0 & 1 & 0\\
0 & \frac{4A_2}{B_2} & 0 & 1
\end{bmatrix}
$. So, by theorem \ref{P5Theo4.8}, we have
\begin{align*}
|P^{-1}\Omega|\geq \frac{(2\pi)^2}{4|B_1B_2|}(1-\epsilon^2)^{\frac{q}{q-2}}\left(\frac{q}{2}\right)^{\left(\frac{4}{q-2}\right)}.
\end{align*}
This gives
\begin{align*}
|\Omega|\geq \frac{(2\pi)^2}{16|B_1B_2|}(1-\epsilon^2)^{\frac{q}{q-2}}\left(\frac{q}{2}\right)^{\left(\frac{4}{q-2}\right)},~\mbox{since}~det(P^{-1})=4.
\end{align*}
This proves equation \eqref{P5eqn52}.
\end{proof}
\subsection{Entropy uncertainty principle}
\begin{theorem}\label{P5Theo4.9}
Let $f\in L^2_\mathbb{H}(\mathbb{R}^2)$ and $g$ be a QWF such that $\|g\|_{L^2_\mathbb{H}(\mathbb{R}^2)}\|f\|_{L^2_\mathbb{H}(\mathbb{R}^2)}=1,$ then
\begin{align}\label{P5eqn54}
\mathcal{E}_{S}(f,g,\wedge_1,\wedge_2)\geq\frac{2}{|B_1B_2|},
\end{align}
where
$\displaystyle\mathcal{E}_{S}(f,g,\wedge_1,\wedge_2)=-\int_{\mathbb{R}^2}\int_{\mathbb{R}^2}\left|\left(\mathcal{S}_{\mathbb{H},g}^{\wedge_1,\wedge_2}f\right)(\boldsymbol x,\boldsymbol\xi)\right|^2\log\left(\left|\left(\mathcal{S}_{\mathbb{H},g}^{\wedge_1,\wedge_2}f\right)(\boldsymbol x,\boldsymbol\xi)\right|^2\right)d\boldsymbol x\boldsymbol d\xi.$
\end{theorem}
\begin{proof}
Define 
\begin{align}\label{P5eqn55}
I(f,g,\wedge_1,\wedge_2,q)=\int_{\mathbb{R}^2}\int_{\mathbb{R}^2}\left|\left(\mathcal{S}_{\mathbb{H},g}^{\wedge_1,\wedge_2}f\right)(\boldsymbol x,\boldsymbol\xi)\right|^qd\boldsymbol xd\boldsymbol\xi.
\end{align}
Then using  \eqref{P5eqn55} in \eqref{P5eqn37}, we get
\begin{align}\label{P5eqn56}
I(f,g,\wedge_1,\wedge_2,2)=\frac{1}{|B_1B_2|}.
\end{align}
Also, from \eqref{P5eqn41} and \eqref{P5eqn56}, it can be shown that
\begin{align}\label{P5eqn57}
I(f,g,\wedge_1,\wedge_2,q)\leq \frac{(2\pi)^{2-q}}{|B_1B_2|}\left(\frac{2}{q}\right)^2.
\end{align}
Define, for $\lambda>0,$
\begin{align*}
R(\lambda)=\frac{I(f,g,\wedge_1,\wedge_2,2)-I(f,g,\wedge_1,\wedge_2,2+2\lambda).}{\lambda}
\end{align*}
Then
\begin{align*}
R(\lambda)
&\geq\frac{1}{\lambda}\left\{\frac{1}{|B_1B_2|}-\frac{(2\pi)^{-2\lambda}}{|B_1B_2|}\left(\frac{1}{1+\lambda}\right)^2\right\}\\
&>\frac{1}{\lambda|B_1B_2|}\left\{1-\frac{1}{(1+\lambda^2)}\right\}
\end{align*}
i.e., 
\begin{align}\label{P5eqn58}
R(\lambda)>\frac{2+\lambda}{|B_1B_2|(1+\lambda)^2}.
\end{align}
Assume that $\mathcal{E}_{S}(f,g,\wedge_1,\wedge_2)<\infty,$ otherwise \eqref{P5eqn54} is obvious.\\
Now using the inequality $1+\lambda \log a\leq a^\lambda,~\lambda>0,$ we have
\begin{align}\label{P5eqn59}
0\leq\frac{1}{\lambda}\left|\left(\mathcal{S}_{\mathbb{H},g}^{\wedge_1,\wedge_2}f\right)(\boldsymbol x,\boldsymbol\xi)\right|^2\left(1-\left|\left(\mathcal{S}_{\mathbb{H},g}^{\wedge_1,\wedge_2}f\right)(\boldsymbol x,\boldsymbol\xi)\right|^{2\lambda}\right)\leq-\left|\left(\mathcal{S}_{\mathbb{H},g}^{\wedge_1,\wedge_2}f\right)(\boldsymbol x,\boldsymbol\xi)\right|^2\log\left(\left|\left(\mathcal{S}_{\mathbb{H},g}^{\wedge_1,\wedge_2}f\right)(\boldsymbol x,\boldsymbol\xi)\right|^2\right).
\end{align}
Since, $-\left|\left(\mathcal{S}_{\mathbb{H},g}^{\wedge_1,\wedge_2}f\right)(\boldsymbol x,\boldsymbol\xi)\right|^2\log\left(\left|\left(\mathcal{S}_{\mathbb{H},g}^{\wedge_1,\wedge_2}f\right)(\boldsymbol x,\boldsymbol\xi)\right|^2\right)$ is integrable, in view of equation \eqref{P5eqn59}, using Lebesgue dominated convergence theorem, we have
\begin{align}\label{P5eqn60}
\lim_{\lambda\rightarrow 0+}R(\lambda)
&=\int_{\mathbb{R}^2}\int_{\mathbb{R}^2}\lim_{\lambda\rightarrow 0+}\left\{\frac{1}{\lambda}\left|\left(\mathcal{S}_{\mathbb{H},g}^{\wedge_1,\wedge_2}f\right)(\boldsymbol x,\boldsymbol\xi)\right|^2\left(1-\left|\left(\mathcal{S}_{\mathbb{H},g}^{\wedge_1,\wedge_2}f\right)(\boldsymbol x,\boldsymbol\xi)\right|^{2\lambda}\right)\right\}d\boldsymbol xd\boldsymbol\xi\notag\\
&=\mathcal{E}_{S}(f,g,\wedge_1,\wedge_2).
\end{align} 
Again from \eqref{P5eqn58}, we get
\begin{align}\label{P5eqn61}
\lim_{\lambda\rightarrow 0+}R(\lambda)\geq\frac{2}{|B_1B_2|}.
\end{align}
Thus from \eqref{P5eqn60} and \eqref{P5eqn61}, we have equation \eqref{P5eqn54}. This completes the proof.
\end{proof}
\begin{corollary}
Let $f\in L^2_\mathbb{H}(\mathbb{R}^2)$ and $g$ be a QWF such that $\|g\|_{L^2_\mathbb{H}(\mathbb{R}^2)}\|f\|_{L^2_\mathbb{H}(\mathbb{R}^2)}=1,$ then
\begin{align}\label{P5eqn62}
\mathcal{E}_{A}(f,g,\wedge_1,\wedge_2)\geq\frac{2}{|B_1B_2|},
\end{align}
where
$\displaystyle\mathcal{E}_{A}(f,g,\wedge_1,\wedge_2)=-\int_{\mathbb{R}^2}\int_{\mathbb{R}^2}\left|\left(\mathcal{A}_\mathbb{H}^{\wedge_1,\wedge_2}(f,g)\right)(\boldsymbol x,\boldsymbol\xi)\right|^2\log\left(\left|\left(\mathcal{A}_\mathbb{H}^{\wedge_1,\wedge_2}(f,g)\right)(\boldsymbol x,\boldsymbol\xi)\right|^2\right)d\boldsymbol x\boldsymbol d\xi.$
\end{corollary}
\begin{proof}
From theorem \ref{P5Theo4.5}, it follows that
$$\left|\left(\mathcal{A}^{\wedge_1,\wedge_2}_{\mathbb{H}}(f,g)\right)(\boldsymbol x,\boldsymbol\xi)\right|=\left|\left(\mathcal{S}^{\wedge_1,\wedge_2}_{\mathbb{H},g}f \right)(\boldsymbol x,\boldsymbol\xi'_{\boldsymbol x})\right|,~\boldsymbol\xi'_{\boldsymbol x}=\left(\xi_1-\frac{A_1x_1}{B_1},\xi_2-\frac{A_2x_2}{B_2}\right).$$
So, we have 
\begin{align*}
\mathcal{E}_{A}(f,g,\wedge_1,\wedge_2)
&=-\int_{\mathbb{R}^2}\int_{\mathbb{R}^2}\left|\left(\mathcal{S}^{\wedge_1,\wedge_2}_{\mathbb{H},g}f \right)(\boldsymbol x,\boldsymbol\xi'_{\boldsymbol x})\right|^2\log\left(\left|\left(\mathcal{S}^{\wedge_1,\wedge_2}_{\mathbb{H},g}f \right)(\boldsymbol x,\boldsymbol\xi'_{\boldsymbol x})\right|^2\right)d\boldsymbol x\boldsymbol d\xi\\
&=\mathcal{E}_{S}(f,g,\wedge_1,\wedge_2).
\end{align*}
Thus using theorem \ref{P5Theo4.9}, we have equation \eqref{P5eqn62}.
\end{proof}
\begin{corollary}
Let $f\in L^2_\mathbb{H}(\mathbb{R}^2)$ and $g$ be a QWF such that $\|g\|_{L^2_\mathbb{H}(\mathbb{R}^2)}\|f\|_{L^2_\mathbb{H}(\mathbb{R}^2)}=1.$ Then
\begin{align}\label{P5eqn54}
\mathcal{E}_{W}(f,g,\wedge_1,\wedge_2)\geq\frac{2-\log 16}{|B_1B_2|},
\end{align}
where
$\displaystyle\mathcal{E}_{W}(f,g,\wedge_1,\wedge_2)=-\int_{\mathbb{R}^2}\int_{\mathbb{R}^2}\left|\left(\mathcal{W}_\mathbb{H}^{\wedge_1,\wedge_2}(f,g)\right)(\boldsymbol x,\boldsymbol\xi)\right|^2\log\left(\left|\left(\mathcal{W}_\mathbb{H}^{\wedge_1,\wedge_2}(f,g)\right)(\boldsymbol x,\boldsymbol\xi)\right|^2\right)d\boldsymbol x\boldsymbol d\xi.$
\end{corollary}
\begin{proof}
From theorem \ref{P5Theo4.6}, it follows that
$$\left|\left(\mathcal{W}^{\wedge_1,\wedge_2}_{\mathbb{H}}(f,g)\right)(\boldsymbol x,\boldsymbol\xi)\right|=4\left|\left(\mathcal{S}^{\wedge_1',\wedge_2'}_{\mathbb{H},\tilde{g}}f \right)(2\boldsymbol x,\boldsymbol\xi'_{\boldsymbol x})\right|,~\boldsymbol\xi'_{\boldsymbol x}=\left(\xi_1-\frac{4A_1x_1}{B_1},\xi_2-\frac{4A_2x_2}{B_2}\right).$$ 
So, we have
\begin{align*}
\mathcal{E}_{W}(f,g,\wedge_1,\wedge_2)
&=-16\int_{\mathbb{R}^2}\int_{\mathbb{R}^2}\left|\left(\mathcal{S}^{\wedge_1',\wedge_2'}_{\mathbb{H},\tilde{g}}f \right)(2\boldsymbol x,\boldsymbol\xi'_{\boldsymbol x})\right|^2\log\left(16\left|\left(\mathcal{S}^{\wedge_1',\wedge_2'}_{\mathbb{H},\tilde{g}}f \right)(2\boldsymbol x,\boldsymbol\xi'_{\boldsymbol x})\right|^2\right)d\boldsymbol x\boldsymbol d\xi\\
&=-4\int_{\mathbb{R}^2}\int_{\mathbb{R}^2}\left|\left(\mathcal{S}^{\wedge_1',\wedge_2'}_{\mathbb{H},\tilde{g}}f \right)(\boldsymbol x,\boldsymbol\xi)\right|^2\log\left(16\left|\left(\mathcal{S}^{\wedge_1',\wedge_2'}_{\mathbb{H},\tilde{g}}f \right)(\boldsymbol x,\boldsymbol\xi)\right|^2\right)d\boldsymbol x\boldsymbol d\xi\\
&=-\frac{4\log 16}{|4B_1B_2|}-4\int_{\mathbb{R}^2}\int_{\mathbb{R}^2}\left|\left(\mathcal{S}^{\wedge_1',\wedge_2'}_{\mathbb{H},\tilde{g}}f \right)(\boldsymbol x,\boldsymbol\xi)\right|^2\log\left(\left|\left(\mathcal{S}^{\wedge_1',\wedge_2'}_{\mathbb{H},\tilde{g}}f \right)(\boldsymbol x,\boldsymbol\xi)\right|^2\right)d\boldsymbol x\boldsymbol d\xi\\
=-\frac{\log 16}{|B_1B_2|}+4\mathcal{E}_{S}(f,\tilde{g},\wedge_1',\wedge_2').
\end{align*}
Therefore, using theorem \ref{P5Theo4.9}, we have
\begin{align*}
\mathcal{E}_{W}(f,g,\wedge_1,\wedge_2)\geq\frac{2-\log 16}{|B_1B_2|}.
\end{align*}
This finishes the proof.
\end{proof}
\section{Conclusions}
In this article, we have studied the Parseval's  identity, sharp Hausdorff-Young inequality for the two sided QQPFT of quaternion valued functions.  Based on the sharp Hausdorff-Young inequality we have obtained the sharper R\`enyi entropy UP for the propose QPFT of quaternion valued functions. We have extended the STQPFT of complex valued functions to the functions of quaternion valued and studied the properties like boundedness, linearity, translation and scaling. We have also obtained the inner product relation and inversion formula for the proposed two sided STQQPFT. We have also obtained the relations of STQQPFT with that of the QQPAF and the QQPWVD of the quaternion valued function associated with the  QQPFT. We have obtained the sharper version of the Lieb's and entropy UPs for all these three transform based on the sharp Hausdorff-Young inequality for the QQPFT.

\section{Acknowledgement}
This work is partially supported by UGC File No. 16-9(June 2017)/2018(NET/CSIR), New Delhi, India.
\bibliography{P5MasterB5_STQQPFT}
\bibliographystyle{plain}
\end{document}